\documentclass[12pt, reqno]{amsart}

\usepackage{amssymb}
\usepackage{mathpazo}
\usepackage{amsmath}%

\usepackage[latin1]{inputenc}
\usepackage{graphicx}
\usepackage[colorlinks=true, breaklinks=true]{hyperref}

\usepackage{amsfonts}

\newtheorem{theorem}{Theorem}

\newtheorem{lemma}[theorem]{Lemma}

\newtheorem{proposition}[theorem]{Proposition}
\newtheorem{remark}[theorem]{Remark}

\newcommand{\bea}{\begin{eqnarray}}
\newcommand{\eq}{\end{eqnarray}}
\newcommand{\eea}{\end{eqnarray}}
\newcommand{\bqn}{\begin{eqnarray*}}
\newcommand{\beaa}{\begin{eqnarray*}}
\newcommand{\eqn}{\end{eqnarray*}}
\newcommand{\eeaa}{\end{eqnarray*}}
\newcommand{\bpr}{\begin{proposition}}
\newcommand{\epr}{\end{proposition}}

\begin{document}
\title{Extreme-strike asymptotics for general Gaussian stochastic volatility
models}
\author{Archil Gulisashvili}
\address{Department of Mathematics, Ohio University, Athens OH 45701, 
\texttt{gulisash@ohio.edu}}
\author{Frederi Viens}
\address{Department of Statistics and Probability, Michigan State
University, East Lansing, MI 48824, \texttt{viens@msu.edu}}
\author{Xin Zhang}
\address{Department of Mathematics, Purdue University, West Lafayette, IN
47907.}
\date{Jan 28, 2017}

\begin{abstract}
We consider a stochastic volatility asset price model in which the
volatility is the absolute value of a continuous Gaussian process with
arbitrary prescribed mean and covariance. By exhibiting a Karhunen-Lo\`{e}ve
expansion for the integrated variance, and using sharp estimates of the
density of a general second-chaos variable, we derive asymptotics for the
asset price density for large or small values of the variable, and study the
wing behavior of the implied volatility in these models. Our main result
provides explicit expressions for the first five terms in the expansion of
the implied volatility. The expressions for the leading three terms are
simple, and based on three basic spectral-type statistics of the Gaussian
process: the top eigenvalue of its covariance operator, the multiplicity of
this eigenvalue, and the $L^{2}$ norm of the projection of the mean function
on the top eigenspace. The fourth term requires knowledge of all
eigen-elements. We present detailed numerics based on realistic liquidity
assumptions in which classical and long-memory volatility models are
calibrated based on our expansion.%
\end{abstract}

\maketitle

\vspace{0.2in}

\vspace{0.2in%
}


\section{Introduction}

In this article, we characterize wing behavior of the implied volatility for
uncorrelated Gaussian stochastic volatility models. This introduction
contains a careful description of the problem's background and of our
motivations. Before going into details, we summarize some of the article's
specificities; all terminology in the next two paragraphs is referenced,
defined, and/or illustrated in the remainder of this introduction.

We hold calibration of volatility smiles as a principal motivator. Cognizant
of the fact that non-centered Gaussian volatility models can be designed in
a flexible and parsimonious fashion, we adopt that class of models, imposing
no further conditions on the marginal distribution of the volatility process
itself, beyond pathwise continuity. The spectral structure of the integrated
variance allows us to work at that level of generality. We find that the
first five terms in the extreme-strike implied volatility asymptotics --
which is typically amply sufficient in applications -- can be determined
explicitly thanks to three parameters characterizing the top of the spectral
decomposition of the integrated variance, with the exception of a factor
appearing in the coefficient of the 4th term in these asymptotics, which
depends on higher-order eigen-elements. In order to prove such a precise
statement while relying on a moderate amount of technicalities, we make use
of the simplifying assumption that the stochastic volatility is independent
of the asset price's driving noise.

When considering the trade-off between this restriction and calibration
considerations, we observe that our model flexibility combined with known
explicit spectral expansions and numerical tools may allow practicioners to
compute the said spectral parameters in a straightforward fashion based on
smile features, while also allowing them to select their favorite Gaussian
volatility model class. Specific examples of Gaussian volatility processes
are non-centered Brownian motion, Brownian bridge, and Ornstein-Uhlenbeck
processes. This last sub-class can be particularly appealing since it
contains stationary volatilities, and includes the well-known Stein-Stein
model. We also mention how any Gaussian model specification, including
long-memory ones, can be handled, thanks to the numerical ability to
determine its spectral elements. We understand that the assumption of the
stochastic volatility model being uncorrelated implies the symmetry of the
implied volatility on either side of the money, which in some applications,
is not a desirable feature. Moreover, while in many option markets,
liquidity considerations limit the ability to calibrate using the
large-strike wing (see the calibration study on SPX options in \cite[Section
5.4]{GJ}), the ability to work with a correlated volatility model is
nonetheless important as soon as one uses the result of the calibration to,
say, price illiquid options such as out-of-the-money calls. Hence a fully
functional general Gaussian model would require a method for estimating the
volatility's correlation with the asset using liquid options data. Such a
study is beyond the scope of our article, since the case of general
correlated Gaussian stochastic volatility models presents additional
mathematical challenges which may require completely new methods and
techniques. We will investigate them separately from this article. An
important step toward a better understanding of the asymptotic behavior of
the implied volatility in some correlated stochastic volatility models is
found in the articles \cite{DFJV1,DFJV2}.

Another problem which is mathematically interesting and important in
practice is the asymptotics for implied volatility in small or large time to
maturity. The techniques developed in the present paper are used in the
subsequent paper \cite{GVZ} to study the small-time asymptotics of
densities, option pricing functions, and the implied volatility in Gaussian
self-similar stochastic volatility models.

\subsection{Background and heuristics}

Studies in quantitative finance based on the Black-Scholes-Merton framework
have shown awareness of the inadequacy of the constant volatility
assumption, particularly after the crash of 1987, when practitioners began
considering that extreme events were more likely than what a log-normal
model will predict. Propositions to exploit this weakness in log-normal
modeling systematically and quantitatively have grown ubiquitous to the
point that implied volatility (IV), or the volatility level that market call
option prices would imply if the Black-Scholes model were underlying, is now
a \emph{bona fide} and vigorous topic of investigation, both at the
theoretical and practical level. The initial evidence against constant
volatility simply came from observing that IV as a function of strike prices
for liquid call options exhibited non-constance, typically illustrated as a
convex curve, often with a minimum near the money as for index options,
hence the term `volatility smile'.

Asset price models where the volatility is a stochastic process are known as
stochastic volatility models; the term `uncorrelated' is added to refer to
the submodel class in which the volatility process is independent of the
noise driving the asset price. In a sense, the existence of the smile for
any uncorrelated stochastic volatility model was first proved mathematically
by Renault and Touzi in \cite{RT}. They established that the IV as a
function of the strike price decreases on the interval where the call is in
the money, increases on the interval where the call is out of the money, and
attains its minimum where the call is at the money. Note that Renault and
Touzi did not prove that the IV is locally convex near the money, but their
work still established stochastic volatility models as a main model class
for studying IV; these models continued steadily to provide inspiration for
IV studies.

A current emphasis, which has become fertile mathematical ground, is on IV
asymptotics, such as large/small-strike, large-maturity, or
small-time-to-maturity behaviors. These are helpful to understand and select
models based on smile shapes. Several techniques are used to derive IV
asymptotics. For instance, by exploiting a method of moments and the
representation of power payoffs as mixtures of a continuum of calls with
varying strikes, in a rather model-free context, R. Lee proved in \cite{Lee}
that, for models with positive moment explosions, the squared IV's large
strike behavior is of order the log-moneyness $\log \left(\frac{K}{%
s_{0}e^{rT}}\right) $ times a constant which depends explicitly on supremum
of the order of finite moments. A similar result holds for models with
negative moment explosions, where the squared IV behaves like $K\mapsto \log
\left(\frac{s_{0}e^{rT}}{K}\right) $ for small values of $K$. More general
formulas describing the asymptotic behavior of the IV in the `wings' ($%
K\rightarrow 0$ or $+\infty $) were obtained in \cite%
{BF1,BF2,BFL,G1,G2,GV,GL} (see also the book \cite{G}).

From the standpoint of modeling, one of the advantages of Lee's original
result is the dependence of IV asymptotics merely on some simple statistics,
namely as we mentioned, in the notation in \cite{Lee}, the maximal order $%
\tilde{p}$ of finite moments for the underlying $S_{T}$, i.e. 
\begin{equation*}
\tilde{p}(T):=\sup \left\{ p\in \mathbb{R}~:~\mathbb{E}\left[ \left(
S_{T}\right) ^{p+1}\right] <\infty \right\} .
\end{equation*}%
This allows the author to draw appropriately strong conclusions about model
calibration. A special class of models in which $\tilde{p}$ is positive and
finite is that of Gaussian volatility models, which we introduce next.

\subsection{Gaussian Stochastic volatility models}

\label{SS:Gsvm}

Let $W$ be a standard Brownian motion on a probability space $(\Omega,%
\mathcal{F},\mathbb{P})$, and let $X$ be a continuous Gaussian process on
the same space that is independent of $W$. We have $X\left( t\right)
=m\left( t\right) +\tilde{X}\left( t\right) $, where $m$ is a continuous
deterministic function on $[0,T]$ (the mean function) and $\tilde{X}$ is a
continuous centered Gaussian process on $[0,T]$ independent of $W$, with
covariance $Q$. Suppose $\{\mathcal{F}_t\}$ is a filtration such that $W$ is
a Brownian motion with respect to $\{\mathcal{F}_t\}$, and the process $X$
is adapted to $\{\mathcal{F}_t\}$.

In the present paper, we study the following asset price model: 
\begin{equation}
dS_{t}=rS_{t}dt+|X_{t}|S_{t}dW_{t}:t\in \lbrack 0,T]  \label{mart}
\end{equation}
on the filtered probability space $(\Omega,\mathcal{F},\{\mathcal{F}_t\},%
\mathbb{P})$, where the filtration $\{\mathcal{F}_t\}$ is such as above. It
is also assumed that the short rate $r$ is constant. The initial condition
for the asset price process will be denoted by $s_{0}$. Note that the
initial condition $X_0$ for the process $X$ may be a nonconstant random
variable.

We will next provide a typical example of a filtration $\{\mathcal{F}_t\}$
satisfying the conditions mentioned above. Let $\mathcal{N}$ be the $\sigma$%
-algebra generated by the events of probability zero, and let $\{\mathcal{F}%
^W_t\}$ and $\{\mathcal{F}^X_t\}$ be the augmentations by the family $%
\mathcal{N}$ of the filtrations generated by the processes $W$ and $X$,
respectively. Consider the filtration $\{\mathcal{F}_t\}$ such that for
every $t\ge 0$, $\mathcal{F}_t=\sigma(\mathcal{F}_t^W,\mathcal{F}^X_t)$.
Then the process $W$ is a Brownian motion with respect to the filtration $\{%
\mathcal{F}_t\}$, and the process $X$ is adapted to $\{\mathcal{F}_t\}$.
Note that if $X_0=const$ a.s., then $\mathcal{F}_0$ is a sub-$\sigma$%
-algebra of $\mathcal{N}$, while if $X_0$ is a random variable, then $%
\mathcal{F}_0=\sigma(X_0;\mathcal{N})$.

Note that it is not supposed in (\ref{mart}) that the process $X$ is a
solution to a stochasic differential equation as is often assumed in
classical stochastic volatility models. A well-known special example of a
Gaussian stochastic volatility model is the Stein-Stein model introduced in 
\cite{SS}, in which the volatility process $X$ is the mean-reverting
Ornstein-Uhlenbeck process satisfying 
\begin{equation}
dX_{t}=\alpha \left( m-X_{t}\right) dt+\beta dZ_{t}  \label{OU}
\end{equation}
where $m$ is the level of mean reversion, $\alpha $ is the mean-reversion
rate, and $\beta $ is level of uncertainty on the volatility; here $Z$ is
another Brownian motion, which may be correlated with $W$. In the present
paper, we adopt an analytic technique, encountered for instance in the
analysis of the uncorrelated Stein-Stein model by this paper's first author
and E.M. Stein in \cite{GS} (see also \cite{G}).

Returning to the question of the value of $\tilde{p}$, for a Gaussian
volatility model, it can sometimes be determined by simple calculations,
which we illustrate here with an elementary example. Assume $S$ is a
geometric Brownian motion with random volatility, i.e. a model as in (\ref%
{mart}) where (abusing notation) $\left\vert X_{t}\right\vert $ is taken the
non-time-dependent $\sigma \left\vert X\right\vert $ where $\sigma $ is a
constant and $X$ is an independent unit-variance normal variate (not
dependent on $t$). Thus, at time $T$, with zero discount rate, $%
S_{T}=s_{0}\exp \left( \sigma \left\vert X\right\vert W_{T}-\sigma
^{2}X^{2}T/2\right) $. To simplify this example to the maximum, also assume
that $X$ is centered; using the independence of $X$ and $W$, we get that we
may replace $\left\vert X\right\vert $ by $X$ in this example, since this
does not change the law of $S_{T}$ (i.e. in the uncorrelated case, $X$'s
non-positivity does not violate standard practice for volatility modeling).
Then, using maturity $T=1$, for any $p>0$, the $p$th moment, via a simple
change of variable, equals%
\begin{equation*}
\mathbb{E}\left[ \left( S_{1}\right) ^{p}\right] =\frac{s_{0}^{p}}{2\pi 
\sqrt{1+p\sigma ^{2}}}\iint_{\mathbf{R}^{2}}dy~dw~\exp \left( -\frac{1}{2}%
\left( y^{2}+w^{2}-2\frac{p\sigma }{\sqrt{1+p\sigma ^{2}}}wy\right) \right)
\end{equation*}%
which by an elementary computation is finite, and equal to $s_{0}^{p}/\sqrt{%
1+p\sigma ^{2}-p^{2}\sigma ^{2}}$, if and only if 
\begin{equation*}
p<\tilde{p}+1=\frac{1}{2}+\sqrt{\frac{1}{4}+\frac{1}{\sigma ^{2}}}.
\end{equation*}%
In the cases where the random volatility model $X$ above is non-centered and
is correlated with $W$, a similar calculation can be performed, at the
essentially trivial expenses of invoking affine changes of variables, and
the linear regression of one normal variate against another.

The above example illustrates heuristically that, by Lee's moment formula,
the computation of $\tilde{p}$ might be the quickest path to obtain the
leading term in the large-strike expansion of the IV, for more complex
Gaussian volatility models, namely ones where the volatility $X$ is
time-dependent. However, computing $\tilde{p}$ is not necessarily an easy
task, and appears, perhaps surprisingly, to have been performed rarely. For
the Stein-Stein model, the value of $\tilde{p}$ can be computed using the
sharp asymptotic formulas for the asset price density near zero and
infinity, established in \cite{GS} for the uncorrelated Stein-Stein model,
and in \cite{DFJV2} for the correlated one. These two papers also provide
asymptotic formulas with error estimates for the IV at extreme strikes in
the Stein-Stein model. Beyond the Stein-Stein model, little was known about
the extreme strike asymptotics of general Gaussian stochastic volatility
models. In the present paper, we extend the above-mentioned results from 
\cite{GS} and \cite{DFJV2} to such models.

\subsection{Motivation and summary of main result}

Adopting the perspective that an asymptotic expansion for the IV can be
helpful for model selection and calibration, our objective is to provide an
expansion for the IV in a Gaussian volatility model relying on a minimal
number of parameters, which can then be chosen to adjust to observed smiles.
The restriction of non-correlated volatility means that the asset price
distribution is a mixture of geometric Brownian motions with time-dependent
volatilities, whose mixing density at time $T$ is that of the square root of
a variable in the second-chaos of a Wiener process. That second-chaos
variable is none other than the integrated variance $\Gamma
_{T}:=\int_{0}^{T}X_{s}^{2}ds. $ By relying on a general Hilbert-space
structure theorem which applies to the second Wiener chaos, we prove that,
for a wide class of non-centered Gaussian stochastic volatility processes
with a possible degeneracy in the eigenstructure of the covariance $Q$ of $X$
viewed as a linear operator on $L^{2}\left( [0,T]\right) $ (i.e. when the
top eigenvalue $\lambda _{1}$ is allowed to have a multiplicity $n_{1}$
larger than $1$), the large-strike IV asymptotics can be expressed with
three terms and an error estimate. These terms depend explicitly on $T$ and
on the following three parameters: $\lambda _{1}$, $n_{1}$, and the ratio $%
\delta =\left\Vert P_{E_{1}}m\right\Vert ^{2}/\lambda _{1}, $ where $%
\left\Vert P_{E_{1}}m\right\Vert $ is the norm in $L^{2}\left( [0,T]\right) $
of the orthogonal projection of the mean function $m$ on the first
eigenspace of $Q$. We also push the expansion to five terms, and notice that
the fifth term also only depends on $\lambda _{1}$, $n_{1}$, and $\delta $,
while the fourth term depends on all other eignevalues and the action of $m$
on all other eigenfunctions. Specifically, with $I\left( K\right) $ the IV
as a function of strike $K$, letting $k:=\log \left( K/s_{0}\right) -rT$ be
the discounted log-moneyness, as $k\rightarrow +\infty $, we prove%
\begin{align}
I\left( K\right) &=M_{1}(T,\lambda _{1})\sqrt{k}+M_{2}\left( T,\lambda
_{1},\delta \right) +M_{3}(T,\lambda _{1},n_{1})\frac{\log k}{\sqrt{k}} 
\notag \\
&\quad+M_{4}(T,\lambda _{1},n_{1},V)\frac{1}{\sqrt{k}}+M_{5}(T,\lambda
_{1},n_{1},\delta )\frac{\log k}{k}+O\left( \frac{1}{\sqrt{k}}\right) ,
\label{asymp}
\end{align}%
where the constants $M_{1}$, $M_{2}$, $M_{3}$, $M_{4}$, and $M_{5}$ depend
explicitly on $T$ and $\lambda _{1}$, $M_{2}$ also depends explicitly on $%
\delta $, while $M_{3}$ also depends explicitly on $n_{1}$, $M_{5}$ depends
explicitly also on both $n_{1}$ and $\delta $, and $M_{4}$ has an additional
rather complex dependence on all the eigen-elements through a factor $V$;
this is all stated in Theorem \ref{T:is} and formula (\ref{E:oo}). A similar
asymptotic formula is obtained in the case where $k\rightarrow -\infty $,
using symmetry properties of uncorrelated stochastic volatility models (see (%
\ref{E:ee})). The specific case of the Stein-Stein model is expanded upon in
some detail. 

\subsection{Practical implications\label{PRACTICAL}}

The first-order constant $M_{1}$ is always strictly positive. The
second-order term (the constant $M_{2}$) vanishes if and only if $m$ is
orthogonal to the first eigenspace of $Q$, which occurs for instance when $%
m\equiv 0$. The third-order and fifth-order terms vanish if and only if the
top eigenvalue has multiplicity $n_{1}=1$, which is typical (the case $%
n_{1}>1$ can be considered degenerate, and does not occur in common
examples). The behavior of $M_{1}$ and $M_{2}$ as functions of $T$ is
determined partly by how the top eigenvalue $\lambda _{1}$ depends on $T$,
which can be non-trivial. In the present paper, we assume $T$ is fixed.

For fixed maturity $T$, assuming that $Q$ has lead multiplicity $n_{1}=1$
for instance, a practitioner will have the possibility of determining a
value $\lambda _{1}$ and a value $\delta $ to match the specific
root-log-moneyness behavior of small- or large-strike IV; moreover in that
case, choosing a constant mean function $m$, one obtains $\delta
=m^{2}\lambda _{1}^{-1}\left\vert \int_{0}^{T}e_{1}\left( t\right)
dt\right\vert ^{2}$ where $e_{1}$ is the top eigenfunction of $Q$. Market
prices may not be sufficiently liquid at extreme strikes to distinguish
between more than two parameters; this is typical of calibration techniques
for implied volatility curves for fixed maturity, such as the `stochastic
volatility inspired' (SVI) parametrization disseminated by J. Gatheral: see 
\cite{Gat,Ga} (see also \cite{GJ} and the references therein). Our result
shows that Gaussian volatility models with non-zero mean are sufficient for
this flexibility, and provide equivalent asymptotics irrespective of the
precise mean function and covariance eigenstructure, since modulo the
disappearance of the third-order term in the unit top multiplicity case $%
n_{1}=1$, only $\lambda _{1}$ and $\delta $ are relevant. The fourth-order
term in our expansion can provide additional precision in calibration. Its
use is illustrated in Section \ref{NUM}.

Modelers wishing to stick to well-known classes of processes for $X$ may
then adjust the value of $\lambda _{1}$ by exploiting any available
invariance properties for the desired class$.$ For example, if $X$ is
standard Brownian motion, or the Brownian bridge, on $[0,T]$, we have $%
\lambda _{1}=4T^{2}/\pi $ or $\lambda _{1}=T^{2}/\pi $ respectively, and
these values scale quadratically with respect to a multiplicative scaling
constant for $X$, beyond which an arbitrary mean value $m$ may be chosen. If 
$X$ is the mean-zero stationary OU process, we have $\lambda _{1}=\beta
^{2}/\left( \omega _{T}+\alpha ^{2}\right) $ where $\omega _{T}$ is the
smallest positive solution of $2\alpha \omega \cos \left( \omega T\right)
+\left( \alpha ^{2}-\omega \right) \sin \left( \omega T\right) =0$, in which
case, for a fixed arbitrarily selected rate of mean reversion $\alpha $, a
scaling of $\lambda _{1}$ is then equivalent to selecting the variance of $X$%
, while a constant mean value $m$ can then be selected independently. \cite[%
Chapter 1]{C} can be consulted for the eigenstructure of the covariance of
Brownian motion and the Brownian bridge, which are classical results, and
for a proof of the eigenstructure of the OU covariance (see also \cite{CP}).
The top eigenfunctions in all three of these cases are known explicit
trigonometric functions (see \cite[Chapter 1]{C}), and need to be referenced
when selecting $m.$ For the OU bridge, the eigenstructure of $Q$
(equivalently known as the Karhunen-Lo\`{e}ve expansion of $Q$) was found in 
\cite{Co}, while in \cite{DM}, such an expansion was characterized for
special Gaussian processes generated by independent pairs of exponential
random variables. On the other hand, fractional Brownian motion and OU
processes driven by fractional Brownian motion (also known as fOU processes)
do not fall in the class of Gaussian processes for which the Karhunen-Lo\`{e}%
ve expansion is known explicitly.

However, efficient numerical techniques allowing to compute the
eigenfunctions and eigenvalues in these cases were developed by S. Corlay
(see Chapter 2 in \cite{C}). Corlay uses the trapezoidal Nystr\"{o}m method
and the three-step Richardson-Romberg method to approximate the five highest
Karhunen-Lo\`{e}ve eigenvalues of various Gaussian processes; in principle,
eigenvalues and eigenfunctions of arbitrarily high order can be obtained
using his method. He starts with such estimates for Brownian motion,
Brownian bridge, and Ornstein-Uhlenbeck process, for which explicit
expressions for the eigenvalues are known. The resulting approximations are
very close to the values obtained from the explicit formulas for the
eigenvalues, which shows that the method used by Corlay is rather powerful.
Corlay also estimates the five highest Karhunen-Lo\`{e}ve eigenvalues of
fractional Brownian motion on $[0,1]$ with the Hurst exponent $H=0.7$. Of
special interest to the context of the present paper is the largest
Karhunen-Lo\`{e}ve eigenvalue $\lambda _{1}$ of fractional Brownian motion,
for which Corlay obtains the approximation $\lambda _{1}\approx
0.374532521757236$.

While we do not need this value, and instead use Corlay's method to compute $%
\lambda _{1}$ for several fOU processes, we are confident that the values we
obtain for the various $\lambda _{1}$'s we use have similar levels of
accuracy to what is illustrated in \cite{C}. Corlay's method is thus one of
the main ingredients in the numerical part of our paper (see the discussion
after (\ref{fixedH}) in Section \ref{NUM}). Fractional OU processes were
proposed early on for option pricing, and recently analyzed in \cite{CR, CV}%
; these processes are versions of the volatility process in the Stein-Stein
model. Therefore, the resulting stochastic volatility models may be called
fractional Stein-Stein models. Section \ref{NUM} illustrates how, in the
case of the classical and fractional Stein-Stein models (OU and fOU
processes), the explicit, semi-explicit, or numerically accessible Karhunen-L%
\`{o}eve expansion of $X$ can be used in conjunction with the asymptotics (%
\ref{asymp}) for calibrating parameters. We find that market liquidity
considerations limit the theoretical range of applicability of calibration
strategies, but that significant practical results are nonetheless available.

\bigskip

The remainder of this article is structured as follows. Section \ref{GEN}
sets up a convenient second-chaos representation for the model's integrated
volatility. In Section \ref{MIX}, we generalize some results from \cite%
{B,H,Z}, concerning the asymptotic behavior of densities of infinite linear
combinations of chi-squared random variables, and derive precise asymptotics
for the density of the mixing distribution. Section \ref{STOCK} converts
these asymptotics into sharp asymptotic formulas for the density of the
asset price $S_{T}$, thanks to the analytic tools developed in \cite{GS, G}.
In Section \ref{IV}, we characterize the wing behavior of the implied
volatility in Gaussian stochastic volatility models. We find sharp
asymptotic formulas for the implied volatility with five explicit terms and
an error estimate. The special case of the uncorrelated Stein-Stein model is
studied in more detail in Section \ref{S:uSS}. Finally, our practical study
of calibration strategies, with numerics, is in Section \ref{NUM}.

\section{General setup and second-chaos expansion of the integrated variance 
\label{GEN}}

\label{S:GC} Let $X$ be an almost-surely continuous Gaussian process on a
filtered complete probability space $(\Omega ,\mathcal{F},\{\mathcal{F}%
_{t}\},\mathbb{P})$ with mean and covariance functions denoted by $m(t)=%
\mathbb{E}[X_{t}]$ and 
\begin{equation*}
Q(t,s)=cov(X_{t},X_{s})=\mathbb{E}\left[ \left( X_{t}-m(t)\right) \left(
X_{s}-m(s)\right) \right] ,
\end{equation*}%
respectively, and suppose the restrictions imposed in (\ref{mart}) are
satisfied. 

Define the centered version of $X$ : $\widetilde{X}_{t}:=X_{t}-m(t)$, $t\geq
0$, and fix a time horizon $T>0$. It is not hard to see that $Q(s,s)>0$ for
all $s>0$. Since the Gaussian process $X$ is almost surely continuous, the
mean function $t\mapsto m(t)$ is a continuous function on $[0,T]$, and the
covariance function $(t,s)\mapsto Q(t,s)$ is a continuous function of two
variables on $[0,T]^{2}$. Indeed, 
the continuity of the process $X$ implies its continuity in probability on $%
\Omega $. Hence, the process $X$ is continuous in the mean-square sense
(see, e.g., \cite{IR}, Lemma 1 on p. 5, or invoke the equivalence of $L^{p}$
norms on Wiener chaos, see \cite{NPbook}). Mean-square continuity of $X$
implies the continuity of the mean function on $[0,T]$. In addition, the
autocorrelation function of the process $X$, that is, the function $R(t,s)=%
\mathbb{E}\left[ X_{t}X_{s}\right] $, $(t,s)\in \lbrack 0,T]^{2}$, is
continuous (see, e.g., \cite{A}, Lemma 4.2). Finally, since $%
Q(t,s)=R(t,s)-m(t)m(s)$, the covariance function $Q$ is continuous on $%
[0,T]^{2}$. We refer the interested reader to \cite{Adler} for more
information on the continuity problems for general Gaussian processes.

In our analysis, it will be convenient to refer to the Karhunen-Lo\`{e}ve
expansion of $\widetilde{X}$. We will next provide certain details
concerning the Karhunen-Lo\`{e}ve expansion and introduce notation that will
be used throughout the paper.

Consider the covariance operator defined by 
\begin{equation*}
\mathcal{K}(f)(t)=\int_{0}^{T}f(s)Q(t,s)ds,\quad f\in L^{2}\left(
[0,T]\right) ,\quad 0\leq t\leq T.
\end{equation*}
The operator $\mathcal{K}$ is a nonnegative compact self-adjoint operator on 
$L^{2}\left( [0,T]\right) $. The non-zero eigenvalues of the operator $%
\mathcal{K}$ are of finite multiplicity, and we assume that they are
rearranged so that 
\begin{equation*}
\lambda _{1}=\lambda _{2}=\dots =\lambda _{n_{1}}>\lambda _{n_{1}+1}=\lambda
_{n_{1}+2}=\dots =\lambda _{n_{1}+n_{2}}>\dots.
\end{equation*}%
In particular, $\lambda _{1}$ is the top eigenvalue, and $n_{1}$ is its
multiplicity. It is known that the series $\sum_{n=1}^{\infty}\lambda_n$
converges. The system of eigenfunctions $E=\{e_{n}\}_{n\ge 1}$,
corresponding to the system $\{\lambda_n\}_{n\ge 1}$, is orthonormal, and
each function $e_n$ is continuous on $[0,T]$. The number $\lambda_0=0$
always belongs to the spectrum of the covariance operator, and it may happen
so that $\lambda_0$ is an eigenvalue of $\mathcal{K}$. The spectral subspace
associated with $\lambda_0$ may be infinite-dimensional, and we choose a
basis $\widetilde{E}$ in this subspace. Then $(E,\widetilde{E})$ is a
complete orthonormal system in $L^{2}\left( [0,T]\right) $. Note that the
eigenvalues and eigenfunctions of $\mathcal{K}$ depend on $T$.

The classical Karhunen-Lo\`{e}ve theorem (see, e.g., \cite{Y}, Section 26.1)
states that there exists an i.i.d. sequence of standard normal variates $%
\left\{ Z_{n}:n=1,2,\ldots \right\} $ such that 
\begin{equation}
\widetilde{X}_{t}=\sum_{n=1}^{\infty }\sqrt{\lambda _{n}}e_{n}(t)Z_{n}.
\label{E:KL}
\end{equation}

\begin{remark}
\label{R:oo} The number of positive eigenvalues may be finite. We will
assume throughout the paper that the set of positive eigenvalues is
infinite; this is the case for all illustrative examples we use, such as the
OU and fOU processes. It is easy to understand how the parameters used in
the paper change if the number of positive eigenvalues is finite.
\end{remark}

Using (\ref{E:KL}), we obtain 
\begin{equation}
\int_{0}^{T}\widetilde{X}_{t}^{2}dt=\int_{0}^{T}\left( \sum_{n=1}^{\infty }%
\sqrt{\lambda _{n}}e_{n}(t)Z_{n}\right) ^{2}dt=\sum_{n=1}^{\infty }\lambda
_{n}Z_{n}^{2}.  \label{E:KLtilde}
\end{equation}%
It is worth pointing out that the previous expression for the integrated
variance in a Gaussian model with centered volatility is in fact the most
general form of a random variable in the second Wiener chaos with
half-bounded support, with mean adjusted to ensure almost-sure positivity of
the integrated variance. This is established using a classical structure
theorem on separable Hilbert spaces, as explained in \cite[Section 2.7.4]%
{NPbook}. In other words (also see \cite[Section 2.7.3]{NPbook} for
additional details), any prescribed mean-adjusted integrated variance in the
second chaos is of the form 
\begin{equation*}
V\left( T\right) :=\iint_{[0,T]^{2}}G\left( s,t\right) dZ\left( s\right)
dZ\left( t\right) +2\left\Vert G\right\Vert _{L^{2}\left( [0,T]^{2}\right)
}^{2}
\end{equation*}%
for some standard Wiener process $Z$ and some function $G\in L^{2}\left(
[0,T]^{2}\right) $. Moreover one can find a centered Gaussian process $%
\tilde{X}$ such that $V\left( T\right) =\int_{0}^{T}\widetilde{X}_{t}^{2}dt$
and one can compute the coefficients $\lambda _{n}$ in the Karhunen-Lo\`{e}%
ve representation (\ref{E:KLtilde}) as the eigenvalues of the covariance of $%
\tilde{X}$.

Let us set 
\begin{equation}
s=\int_{0}^{T}m(t)^{2}dt\quad\mbox{and}\quad\delta
_{n}=\int_{0}^{T}m(t)e_{n}(t)dt,\quad n\geq 1.  \label{E:delt}
\end{equation}
Then, it follows from (\ref{E:KLtilde}) and (\ref{E:delt}) that, for the
non-centered process $X$, 
\begin{align}
\int_{0}^{T}X_{t}^{2}dt&=\sum_{n=1}^{\infty }\lambda
_{n}Z_{n}^{2}+2\sum_{n=1}^{\infty }\sqrt{\lambda _{n}}\delta_n Z_{n}+s 
\notag \\
&=\sum_{n=1}^{\infty }\lambda _{n}\left[ Z_{n}+\frac{\delta _{n}}{\sqrt{
\lambda _{n}}}\right] ^{2}+\left( s-\sum_{n=1}^{\infty }\delta
_{n}^{2}\right).  \label{E:kl}
\end{align}

\begin{remark}
\label{R:1} It is easy to see, using (\ref{E:kl}) that if the function $%
t\mapsto m(t)$ belongs to the subspace of $L^{2}[0,T]$ generated by the
orthonormal system $E$, then 
\begin{equation}
\int_{0}^{T}X_{t}^{2}dt=\sum_{n=1}^{\infty }\lambda _{n}\left[ Z_{n}+\frac{%
\delta _{n}}{\sqrt{\lambda _{n}}}\right] ^{2}.  \label{E:eco1}
\end{equation}
For instance, the equality in (\ref{E:eco1}) holds if $\lambda=0$ is not an
eigenvalue of the operator $\mathcal{K}$. In the case where the process $X$
is centered, we have 
\begin{equation}
\int_{0}^{T}X_{t}^{2}dt=\sum_{n=1}^{\infty }\lambda _{n} Z_{n}^{2}.
\label{E:eco2}
\end{equation}
Note that the right-hand sides of (\ref{E:eco1}) and (\ref{E:eco2}) are
infinite linear combinations of chi-square random variables.
\end{remark}

Let us denote the chi-squared distribution with the number of degrees of
freedom $k$ and the parameter of noncentrality $\lambda $ by $\chi
^{2}(k,\lambda )$ (more information on such distributions can be found in 
\cite{G} or in any probability textbook; the convention used here is that
the mean of $\chi ^{2}(k,\lambda )$ is $k+\lambda $). Set 
\begin{equation}
\Lambda _{T}=\frac{1}{\lambda _{1}}\left(
\int_{0}^{T}X_{t}^{2}dt-s+\sum_{n=1}^{\infty }\delta _{n}^{2}\right)
\label{E:ella}
\end{equation}%
and denote 
\begin{equation}
\xi _{0}=\sum_{n=1}^{n_{1}}\delta _{n}^{2};\quad \xi
_{k}=\sum_{n=n_{1}+\cdots +n_{k}+1}^{n_{1}+\cdots +n_{k+1}}\delta
_{n}^{2},\quad k\geq 1.  \label{E:labb}
\end{equation}%
Denote also 
\begin{equation}
\delta =\frac{\xi _{0}}{\lambda _{1}}.  \label{E:deli}
\end{equation}%
Then, it is not hard to see, using (\ref{E:kl}), (\ref{E:ella}), (\ref%
{E:labb}), and (\ref{E:deli}), that 
\begin{equation}
\Lambda _{T}=\chi ^{2}\left( n_{1},\delta \right) +\sum_{k=2}^{\infty }\frac{%
\lambda _{n_{1}+\cdots +n_{k-1}+1}}{\lambda _{1}}\chi ^{2}\left( n_{k},\frac{%
1}{\lambda _{n_{1}+\cdots +n_{k-1}+1}}\xi _{k-1}\right) ,  \label{E:kl4}
\end{equation}%
where the repeated chi-squared notation is used abusively to denote
independent chi-square random variables. We will denote the distribution
density of $\Lambda _{T}$ by $q_{T}$.

\section{Asymptotics of the mixing density\label{MIX}}

The asymptotic behavior of the distribution density of an infinite linear
combination of independent central chi-squared random variables was
characterized by Zolotarev (see formula (5) in \cite{Z}). In \cite{H},
Hoeffding found more general and sharp formulas. The results obtained by
Zolotarev and Hoeffding were generalized to the case of noncentral
chi-squared variables by Beran (see \cite{B}). Note that Beran considered
infinite sums of chi-squared variables with all the noncentrality parameters
strictly greater than zero. Since there is a gap betweed the results of
Zolotarev, Hoeffding, and Beran, we decided to include a discussion of a
similar result, where there are no restrictions on the noncentrality
parameters. Keeping in mind the series in (\ref{E:kl4}), we will study the
asymptotic behavior of the density $q$ of the following infinite sum: 
\begin{equation}
\Lambda=\chi ^{2}\left( n_{1},\eta_1\right) +\sum_{k=2}^{\infty}\rho_k\chi
^{2}\left( n_k, \eta_k\right),  \label{E:klm4}
\end{equation}
where $n_k\ge 1$, $k\ge 1$, are integers, and $\eta_k\ge 0$ for all $k\ge 2$%
. If $\eta_k=0$ for some $k$, then the corresponding chi-squared random
variable is central. It is also assumed that $1>\rho_2>\rho_3>\cdots> 0$, 
\begin{equation}
\sum_{k=2}^{\infty}n_k\rho_k<\infty,\quad\sum_{k=2}^{\infty}\eta_k\rho_k<%
\infty,  \label{E:ext}
\end{equation}
and the chi-squared random variables in (\ref{E:klm4}) are independent. We
will denote by $q_{\Lambda}$ the distribution density of the random variable 
$\Lambda$.

The distribution density of a chi-squared random variable $\chi ^{2}(n,\eta)$
will be denoted by $p_{\chi^2}(\cdot;n,\eta)$. It is known that if $\eta> 0$%
, then 
\begin{equation}
p_{\chi^2}(x;n,\eta)=\frac{1}{2}\left(\frac{x}{\eta}\right)^{\frac{n}{4}-%
\frac{1}{2}} e^{-\frac{x+\eta}{2}}I_{\frac{n}{2}-1}(\sqrt{\eta x}),\quad x>
0,  \label{E:gell}
\end{equation}
where $I_{\nu}$ is the modified Bessel function of the first kind (see,
e.g., \cite{G}, Theorem 1.31). For $\eta=0$, we have 
\begin{equation}
p_{\chi^2}(x;n,0)=\frac{1}{2^{\frac{n}{2}}\Gamma\left(\frac{n}{2}\right)}x^{%
\frac{n-2}{2}}\exp\left\{-\frac{x}{2} \right\},\quad x> 0  \label{E:gelli}
\end{equation}
(see, e.g., Lemma 1.27 in \cite{G}). It is not hard to see that $%
\lim_{\eta\rightarrow 0}p_{\chi^2}(x;n,\eta)=p_{\chi^2}(x;n,0). $ Let us
also mention that 
\begin{equation}
I_{\nu}(t)=\frac{e^t}{\sqrt{2\pi t}}\left(1+O\left(t^{-1}\right)\right),%
\quad t\rightarrow\infty,  \label{E:56}
\end{equation}
for all $\nu\ge 0$ (see, e.g., 9.6.7 in \cite{AS}).

It is known that for $t<\frac{1}{2}$, the moment generating function of a
chi-squared random variable $\chi ^{2}(n,\eta )$ with $\eta \geq 0$ is as
follows: 
\begin{equation}
t\mapsto \frac{1}{(1-2t)^{\frac{n}{2}}}\exp \left\{ \frac{\eta t}{1-2t}%
\right\} .  \label{E:mgf}
\end{equation}%
In the formulation of the next result, we will use the following number: 
\begin{equation*}
A=\mathbb{E}\left[ \exp \left\{ \frac{U}{2}\right\} \right] ,
\end{equation*}
where $U$ is defined as $\Lambda $ without the first term:%
\begin{equation}
U=\sum_{k=2}^{\infty }\rho _{k}\chi ^{2}\left( n_{k},\eta _{k}\right) .
\label{E:set}
\end{equation}%
Next, using (\ref{E:set}) and (\ref{E:mgf}), we obtain 
\begin{equation}
A=\prod_{k\geq 2}(1-\rho _{k})^{-\frac{n_{k}}{2}}\exp \left\{ \frac{\eta
_{k}\rho _{k}}{2(1-\rho _{k})}\right\} ,  \label{E:Aa}
\end{equation}%
and it is not hard to see, by taking into account (\ref{E:ext}), that $%
A<\infty $.

The next assertion is based on the results of Zolotarev, Hoeffding, and
Beran.

\begin{theorem}
\label{T:zhb} Suppose the conditions formulated after formula (\ref{E:klm4})
hold. If $\eta_1> 0$, then 
\begin{equation}
\left|\frac{q_{\Lambda}(x)}{p_{\chi^2}\left(x;n_1,\eta_1\right)}%
-A\right|=O\left(x^{-\frac{1}{2}}\right)  \label{E:forma1}
\end{equation}
as $x\rightarrow\infty$, while if $\eta_1=0$, then 
\begin{equation}
\left|\frac{q_{\Lambda}(x)}{p_{\chi^2}\left(x;n_1,0\right)}%
-A\right|=O\left(x^{-1}\right)  \label{E:forma2}
\end{equation}
as $x\rightarrow\infty$. In the formulas above, the constant $A$ is given by
(\ref{E:Aa}).
\end{theorem}

\begin{remark}
\label{R:HZ} Theorem \ref{T:zhb} is a minor generalization of similar
propositions obtained in \cite{H} and \cite{B}. The difference between those
propositions and our Theorem 3 is that \cite{H} assumes that all the
chi-squared variables in (\ref{E:klm4}) are central, in Theorem 2 in \cite{B}
they are all assumed noncentral, while in our Theorem \ref{T:zhb}, we may
have any combination of central and non-central chi-squared variables.

Theorem 2 in \cite{B} provides an asymptotic formula for the complementary
distribution function (tail) of an infinite linear combination of
independent noncentral chi-square random variables. A sharper formula for
the distribution density of such a linear combination can be extracted from
the proof of Theorem 2 in \cite{B} (see the very end of that proof).
\end{remark}

\emph{Sketch of the proof of Theorem \ref{T:zhb}.} We follow the proof of
Theorem 2 in \cite{B}. Let us denote by $p_{U}$ the distribution density of
the random variable $U$ in (\ref{E:set}). Then 
\begin{equation}
q_{\Lambda }(x)=\int_{0}^{x}p_{\chi ^{2}}\left( x-y;n_{1},\eta _{1}\right)
p_{U}(y)dy,\quad x>0.  \label{E:lab}
\end{equation}%
Let us fix $0<\alpha <1$. We have 
\begin{equation*}
\frac{q_{\Lambda }(x)}{p_{\chi ^{2}}\left( x;n_{1},\eta _{1}\right) }%
-A=V_{1}+V_{2}+V_{3}+V_{4},
\end{equation*}%
where 
\begin{align*}
& V_{1}=\int_{0}^{\alpha x}\left[ \left( 1-\frac{y}{x}\right) ^{\frac{n_{1}}{%
2}-1}-1\right] W(x,y)\exp \left\{ \frac{y}{2}\right\} p_{U}(y)dy, \\
& V_{2}=\int_{\alpha x}^{x}\left( 1-\frac{y}{x}\right) ^{\frac{n_{1}}{2}%
-1}W(x,y)\exp \left\{ \frac{y}{2}\right\} p_{U}(y)dy, \\
& V_{3}=\int_{0}^{\alpha x}[W(x,y)-1]\exp \left\{ \frac{y}{2}\right\}
p_{U}(y)dy, \\
& V_{4}=-\int_{\alpha x}^{\infty }\exp \left\{ \frac{y}{2}\right\}
p_{U}(y)dy.
\end{align*}%
In the formulas above, the function $W$ is defined by 
\begin{equation*}
W(x,y)=\left( 1-\frac{y}{x}\right) ^{-\frac{n_{1}}{4}+\frac{1}{2}}\frac{I_{%
\frac{n}{2}-1}(\sqrt{\eta (x-y)})}{I_{\frac{n}{2}-1}(\sqrt{\eta x})}
\end{equation*}%
if $\eta _{1}>0$, while if $\eta _{1}=0$, then $W(x,y)=1$. Note that $\eta
_{1}=0$ implies $V_{3}=0$. Then, using calculations similar to those in the
proof of Theorem 2 in \cite{B}, we find that when $\eta _{1}>0$, $V_{3}$ is
the leading term and is of order $x^{-1/2}$, while when $\eta _{1}=0$, this
term vanishes, and the next highest-order term is of order $x^{-1}$. This
explains the different error estimates in the formulas in Theorem \ref{T:zhb}%
. We include two auxiliary statements below (Lemmas \ref{L:ue} and \ref%
{L:ues}). They are needed to perform the above-mentioned calculations. This
finishes the sketch of the proof of Theorem \ref{T:zhb}. \hspace*{\fill}$%
\square \vspace{0.1in}$

\begin{lemma}
\label{L:ue} Under the assumptions in Theorem \ref{T:zhb}, the following
holds: 
\begin{equation*}
\mathbb{E}\left[ U\exp \left\{ \frac{U}{2}\right\} \right] <\infty .
\end{equation*}
\end{lemma}

\emph{Proof.} This follows in a straightforward way (details omitted), using
(\ref{E:set}), differentiating the function in (\ref{E:mgf}), and taking
into account the resulting formula and (\ref{E:Aa}), implying that: 
\begin{equation*}
\mathbb{E}\left[ U\exp \left\{ \frac{U}{2}\right\} \right]
=A\sum_{k=2}^{\infty }\rho _{k}\left[ \frac{n_{k}}{1-\rho _{k}}+\frac{\eta
_{k}}{(1-\rho _{k})^{2}}\right]
\end{equation*}%
so that that Lemma \ref{L:ue} clearly follows from (\ref{E:ext}) and the
finiteness $A<\infty $. \hspace*{\fill}$\square \vspace{0.1in}$

\begin{lemma}
\label{L:ues} Under the restrictions in Theorem \ref{T:zhb}, there exists a
number $\varepsilon> 0$, depending on the constants in (\ref{E:set}), and
such that 
\begin{equation*}
p_U(y)=O\left(\exp\left\{-\left(\frac{1}{2}+\varepsilon\right)y\right\}%
\right)
\end{equation*}
as $y\rightarrow\infty$.
\end{lemma}

\emph{Proof.} We have $U=\rho_2\widetilde{U}$, where $\widetilde{U}%
=\sum_{k=2}^{\infty}\widetilde{\rho}_k\chi ^{2}\left( n_k, \eta_k\right) $
with $\widetilde{\rho}_2=1$ and $\widetilde{\rho}_k=\frac{\rho_k}{\rho_2}$
for all $k\ge 3$. It follows that $p_U(x)=\frac{1}{\rho_2}p_{\widetilde{U}%
}\left(\frac{1}{\rho_2}y\right)$. Since $\rho_2< 1$, and the random varaible 
$\widetilde{U}$ has the same structure as the random variable $\Lambda$ in (%
\ref{E:klm4}), it suffices to show that for every $\tau> 0$, 
\begin{equation}
q_{\Lambda}(x)=O\left(\exp\left\{\left(-\frac{1}{2}+\tau\right)y\right\}%
\right)  \label{E:gene}
\end{equation}
as $x\rightarrow\infty$.

Let us first assume $n_{1}\geq 2$. Then, using (\ref{E:lab}), (\ref{E:gell}%
), the fact that the function $I_{\nu }$ is increasing, and (\ref{E:56}), we
obtain (\ref{E:gene}). Next, let $n_{1}=1$. We have 
\begin{equation*}
\Lambda \leq \chi ^{2}\left( n_{1},\eta _{1}\right) +\chi ^{2}\left(
n_{2},\eta _{2}\right) +\sum_{k=3}^{\infty }\rho _{k}\chi ^{2}\left(
n_{k},\eta _{k}\right) .
\end{equation*}%
Next, we observe that $\chi ^{2}\left( n_{1},\eta _{1}\right) +\chi
^{2}\left( n_{2},\eta _{2}\right) =\chi ^{2}\left( n_{1}+n_{2},\eta
_{1}+\eta _{2}\right) $ (the previous formula follows from (\ref{E:mgf})).
This reduces the case where $n_{1}=1$ to the already considered case where $%
n_{1}>1$. It follows from the previous reasoning that (\ref{E:gene}) holds.
This completes the proof of Lemma \ref{L:ues}. \hspace*{\fill}$\square 
\vspace{0.1in}$

Theorem \ref{T:zhb} will allow us to characterize the asymptotic behavior of
the distribution density $q_{T}$ of the random variable $\Lambda _{T}$
defined by (\ref{E:kl4}). Using Theorem \ref{T:zhb}, we see that if $\delta
>0$, then 
\begin{equation}
\left\vert \frac{q_{T}(x)}{p_{\chi ^{2}}\left( x;n_{1},\delta \right) }%
-A\right\vert =O\left( x^{-\frac{1}{2}}\right)  \label{E:del}
\end{equation}%
as $x\rightarrow \infty $, while if $\delta =0$, then 
\begin{equation}
\left\vert \frac{q_{T}(x)}{p_{\chi ^{2}}\left( x;n_{1},0\right) }%
-A\right\vert =O\left( x^{-1}\right)  \label{E:delic}
\end{equation}%
as $x\rightarrow \infty $. In (\ref{E:del}) and (\ref{E:delic}), the formula
for $A$ is%
\begin{equation}
A=\prod_{j>n_{1}}\left( \frac{\lambda _{1}}{\lambda _{1}-\lambda _{j}}%
\right) ^{\frac{1}{2}}\exp \left\{ \frac{\delta _{j}^{2}}{2(\lambda
_{1}-\lambda _{j})}\right\} .  \label{E:della}
\end{equation}

It is clear that for $\delta> 0$, (\ref{E:del}) gives 
\begin{equation}
q_T(x)=Ap_{\chi^2}\left(x;n_1,\delta\right)\left(1+O\left(x^{-\frac{1}{2}%
}\right)\right)  \label{E:dellas}
\end{equation}
as $x\rightarrow\infty$. Similarly, for $\delta=0$, (\ref{E:delic}) implies
that 
\begin{equation}
q_T(x)=Ap_{\chi^2}\left(x;n_1,0\right)\left(1+O\left(x^{-1}\right)\right)
\label{E:dellus}
\end{equation}
as $x\rightarrow\infty$.

It is known that 
\begin{equation*}
I_{\nu}(t)=\frac{e^t}{\sqrt{2\pi t}}\left(1+O\left(t^{-1}\right)\right)
\quad t\rightarrow\infty,
\end{equation*}
(see 9.7.1 in \cite{AS}). Next, using the previous formula in (\ref{E:gell}%
), we obtain 
\begin{equation}
p_{\chi^2}(x;n,\lambda)=\frac{1}{2\sqrt{2\pi}}\lambda^{-\frac{n-1}{4}}x^{%
\frac{n-3}{4}} e^{\sqrt{\lambda x}}e^{-\frac{x+\lambda}{2}%
}\left(1+O\left(x^{-\frac{1}{2}}\right)\right)  \label{E:gellla}
\end{equation}
as $x\rightarrow\infty$.

Recall that we denoted by $q_T$ the distribution density of the random
variable $\Lambda_T$ defined by (\ref{E:ella}). Using (\ref{E:dellas}) and (%
\ref{E:gellla}), we see that for $\delta> 0$, 
\begin{equation}
q_T(x)=\frac{A}{2\sqrt{2\pi}}\delta^{-\frac{n_1-1}{4}}x^{\frac{n_1-3}{4}} e^{%
\sqrt{\delta x}}e^{-\frac{x+\delta}{2}}\left(1+O\left(x^{-\frac{1}{2}%
}\right)\right)  \label{E:qu}
\end{equation}
as $x\rightarrow\infty$. The constants $A$ and $\delta$ in (\ref{E:qu}) are
defined by (\ref{E:della}) and (\ref{E:deli}), respectively.

We next turn our attention to the case where $\delta =0$. In this case, it
follows from (\ref{E:dellus}), (\ref{E:della}), and (\ref{E:gelli}) that 
\begin{align}
q_{T}(x)& =\frac{1}{2^{\frac{n_{1}}{2}}\Gamma \left( \frac{n_{1}}{2}\right) }%
\prod_{k>n_{1}}\left( \frac{\lambda _{1}}{\lambda _{1}-\lambda _{k}}\right)
^{\frac{1}{2}}\exp \left\{ \frac{\delta _{k}^{2}}{2(\lambda _{1}-\lambda
_{k})}\right\} x^{\frac{n_{1}-2}{2}}\exp \left\{ -\frac{x}{2}\right\}  \notag
\\
& \quad \times \left( 1+O\left( x^{-1}\right) \right)  \label{E:qus}
\end{align}%
as $x\rightarrow \infty $.

\begin{remark}
In comparing (\ref{E:qu}) and (\ref{E:qus}), one notes that the latter
cannot be obtained from the former by letting $\delta $ tend to $0$: while
the exponential terms would match, the power terms do not, and an additional
discrepancy would occur when $n_{1}>1$ from the singular term $\delta
^{-(n_{1}-1)/4}.$
\end{remark}

Our next goal is to characterize the asymptotic behavior of the distribution
density $p_T$ of the integrated variance $\Gamma_T=\int_0^TX_t^2dt$. The
following statement holds.

\begin{theorem}
\label{T:gella} (i)\,\,If $\delta> 0$, then 
\begin{align}
p_T(x)&=Cx^{\frac{n_1-3}{4}}\exp\left\{\sqrt{\frac{\delta}{\lambda_1}}\sqrt{x%
}\right\} \exp\left\{-\frac{x}{2\lambda_1}\right\}\left(1+O\left(x^{-\frac{1%
}{2}}\right)\right)  \label{E:finn}
\end{align}
as $x\rightarrow\infty$, where 
\begin{align}
C&=\frac{1}{2\sqrt{2\pi}}\lambda_1^{-\frac{n_1+1}{4}}\delta^{-\frac{n_1-1}{4}%
} \exp\left\{\frac{s-\sum_{n=1}^{\infty} \delta_n^2} {2\lambda_1}-\frac{%
\delta}{2}\right\}  \notag \\
&\quad\times\prod_{j> n_1}^{\infty}\left(\frac{\lambda_1}{\lambda_1-\lambda_j%
}\right)^{\frac{1}{2}}\exp\left\{\frac{\delta_j^2} {2(\lambda_1-\lambda_j)}%
\right\}.  \label{E:rr}
\end{align}
(ii)\,\,If $\delta=0$, then 
\begin{align}
p_T(x)&=Cx^{\frac{n_1-2}{2}}\exp\left\{-\frac{x}{2\lambda_1}\right\}
\left(1+O\left(x^{-1}\right)\right)  \label{E:finna}
\end{align}
as $x\rightarrow\infty$, where 
\begin{equation}
C=\frac{1}{2^{\frac{n_1}{2}}\Gamma\left(\frac{n_1}{2}\right)}\lambda_1^{-%
\frac{n_1}{2}} \exp\left\{\frac{s-\sum_{n> n_1}\delta_n^2} {2\lambda_1}%
\right\} \prod_{k> n_1}\left(\frac{\lambda_1}{\lambda_1-\lambda_k}\right)^{%
\frac{1}{2}}\exp\left\{\frac{\delta_j^2} {2(\lambda_1-\lambda_j)}\right\}.
\label{E:ho}
\end{equation}
In particular, if the process $X$ is centered, then (\ref{E:finna}) holds
with 
\begin{equation}
C=\frac{1}{2^{\frac{n_1}{2}}\Gamma\left(\frac{n_1}{2}\right)}\lambda_1^{-%
\frac{n_1}{2}} \prod_{k> n_1}\left(\frac{\lambda_1}{\lambda_1-\lambda_k}%
\right)^{\frac{1}{2}}.  \label{E:hor}
\end{equation}
\end{theorem}

\emph{Proof.}\ It follows from (\ref{E:ella}) that $p_{T}(x)=\frac{1}{%
\lambda _{1}}q_{T}\left( \frac{1}{\lambda _{1}}(x-\tau )\right) ,$ where $%
\tau =s-\sum_{n=1}^{\infty }\delta _{n}^{2}$. Now, formula (\ref{E:qu})
implies that 
\begin{align}
p_{T}(x)& =\frac{A}{2\sqrt{2\pi }}\frac{1}{\lambda _{1}}\lambda _{1}^{\frac{%
n_{1}-1}{4}}\delta ^{-\frac{n_{1}-1}{4}}  \notag \\
& \quad \times \left( \sum_{n=1}^{n_{1}}\delta _{n}^{2}\right) ^{-\frac{%
n_{1}-1}{4}}\lambda _{1}^{-\frac{n_{1}-3}{4}}\exp \left\{ \frac{\tau
-\sum_{n=1}^{n_{1}}\delta _{n}^{2}}{2\lambda _{1}}\right\}  \notag \\
& \quad \times (x-\tau )^{\frac{n_{1}-3}{4}}\exp \left\{ \sqrt{\frac{\delta
(x-\tau )}{\lambda _{1}}}\right\} \exp \left\{ -\frac{x}{2\lambda _{1}}%
\right\}  \notag \\
& \quad \times \left( 1+O\left( x^{-\frac{1}{2}}\right) \right)  \label{E:ah}
\end{align}%
as $x\rightarrow \infty $.

Next, taking into account that 
\begin{equation*}
(x-\tau )^{\frac{n_{1}-3}{4}}=x^{\frac{n_{1}-3}{4}}(1+O(x^{-1}))
\end{equation*}%
and 
\begin{equation*}
\exp \left\{ \sqrt{\frac{\delta (x-\tau )}{\lambda _{1}}}\right\} =\exp
\left\{ \sqrt{\frac{\delta }{\lambda _{1}}}\sqrt{x}\right\} (1+O(x^{-\frac{1%
}{2}})),
\end{equation*}%
and simplifying the expression on the right-hand side of (\ref{E:ah}), we
obtain (\ref{E:finn}). The proof of formula (\ref{E:finna}) is similar,
using (\ref{E:qus}). \hspace*{\fill}$\square \vspace{0.1in}$

\section{Asset price asymptotics\label{STOCK}}

\label{S:Asy} The model in (\ref{mart}) is described by a linear stochastic
differential equation. Therefore, we have 
\begin{equation}
S_t=s_0\exp\left\{rt-\frac{1}{2}\int_0^tX_s^2ds+\int_0^t|X_s|dW_s\right\}.
\label{E:DD}
\end{equation}
The previous equality can be derived from the Dol\'{e}ans-Dade formula (see 
\cite{RY}). Since the processes $X$ and $W$ are independent, the following
formula holds for the distribution density $D_t$ of the asset price $S_t$: 
\begin{align}
&D_t(x)=\frac{\sqrt{s_0e^{rt}}}{\sqrt{2\pi t}}x^{-\frac{3}{2}}
\int_0^{\infty} y^{-1}\exp\left\{-\left[\frac{\log^2\frac{x}{s_0e^{rt}}}{%
2ty^2}+\frac{ty^2}{8}\right]\right\}\widetilde{p}_t(y)dy.  \label{E:sic2}
\end{align}
In (\ref{E:sic2}), $\widetilde{p}_t$ is the distribution density of the
random variable $\widetilde{Y}_t=\left\{\frac{1}{t}\int_0^tX^2_sds\right\}^{%
\frac{1}{2}}. $ The function $\widetilde{p}_t$ is called the mixing density.
The proof of formula (\ref{E:sic2}) can be found in \cite{G} (see (3.5) in 
\cite{G}). It is not hard to see that $\widetilde{p}_t(y)=2typ_t\left(ty^2%
\right), $ where the symbol $p_t$ stands for the density of the realized
volatility $Y_t=\int_0^tX^2_sds$.

Suppose first that the volatility process is such that $\delta> 0$. It
follows from formula (\ref{E:finn}) that 
\begin{align}
&\widetilde{p}_t(y)=\widetilde{A}y^{\frac{n_1-1}{2}} \exp\left\{\widetilde{B}%
y\right\} \exp\left\{-\widetilde{C}y^2\right\}
\left(1+O\left(y^{-1}\right)\right)  \label{E:finchi}
\end{align}
as $y\rightarrow\infty$, where 
\begin{equation}
\widetilde{A}=2Ct^{\frac{n_1+1}{4}}, \quad\widetilde{B}=\sqrt{\frac{\delta t%
}{\lambda_1}},\quad\widetilde{C}=\frac{t}{2\lambda_1}.  \label{E:fitch}
\end{equation}
In (\ref{E:fitch}), the constant $C$ is defined by (\ref{E:rr}).

Our next goal is to estimate the function $D_t$. The asymptotic behavior as $%
x\rightarrow\infty$ of the integral appearing in (\ref{E:sic2}) was studied
in \cite{GS} (see also Section 5.3 in \cite{G}). It is explained in \cite{G}
how to get an asymptotic formula for the integral in (\ref{E:sic2}) in the
case where an asymptotic formula for the mixing density is similar to
formula (\ref{E:finchi}). We refer the reader to the derivation of Theorem
6.1 in \cite{G}, which is based on formula (5.133) in Section 5.6 of \cite{G}
and Theorem 5.5 in \cite{G}. The latter theorem concerns the asymptotic
behavior of integrals with lognormal kernels. Having obtained an asymptotic
formula for the distribution density of the asset price, we can find a
similar asymptotic formula for the call pricing function $C$ at large
strikes, and then obtain an asymptotic formula for the implied volatility $I$
(see Section 10.5 in \cite{G}).

Theorem 5.5 in \cite{G} provides an asymptotic formula as $%
w\rightarrow\infty $ for the integral 
\begin{equation*}
\int_0^{\infty}A(y)\exp\left\{-\left(\frac{w^2}{y^2}+k^2y^2\right)\right\}dy,
\end{equation*}
where $k> 0$ is fixed, and it is assumed that 
\begin{equation*}
A(y)=e^{ly}\zeta(y)(1+O(b(y)))
\end{equation*}
as $y\rightarrow\infty$. In the previous asymptotic formula, $l$ is a real
number, and $\zeta$ and $b$ are functions satisfying certain conditions.

Let us fix $T> 0$. Our goal is to use Theorem 5.5 in \cite{G} with 
\begin{equation*}
A(y)=y^{-1}\widetilde{p}_T(y)\exp\left\{\widetilde{C}y^2\right\},
\end{equation*}
$l=\widetilde{B}$, $\zeta(y)=\widetilde{A}y^{\frac{n_1-3}{2}}$, $b(y)=y^{-1}$%
, $w=(2T)^{-\frac{1}{2}}\log\frac{x}{s_0e^{rT}}$, $k=\frac{\sqrt{8\widetilde{%
C}+T}}{2\sqrt{2}}$, and $\gamma=1$ (see the formulation of Theorem 5.5 in 
\cite{G} for the meaning of the constant $\gamma$). This gives 
\begin{align}
&\int_0^{\infty} y^{-1}\exp\left\{-\left[\frac{\log^2\frac{x}{s_0e^{rT}}}{%
2ty^2}+\frac{Ty^2}{8}\right]\right\}\widetilde{p}_T(y)dy=\frac{\widetilde{A}%
2^{\frac{n_1-1}{4}}\sqrt{\pi}} {T^{\frac{n_1-3}{8}}(8\widetilde{C}+T)^{\frac{%
n_1+1}{8}}}\left(s_0e^{rT}\right) ^{\frac{\sqrt{8\widetilde{C}+T}}{2\sqrt{T}}%
}  \notag \\
&\exp\left\{\frac{\widetilde{B}^2}{2(8\widetilde{C}+T)}\right\}\left(\log%
\frac{x}{s_0e^{rT}}\right)^{\frac{n_1-3}{4}}x^{-\frac{\sqrt{8\widetilde{C}+T}%
}{2\sqrt{T}}}\exp\left\{\frac{\widetilde{B}\sqrt{2}}{T^{\frac{1}{4}}(8%
\widetilde{C}+T)^{\frac{1}{4}}}\sqrt{\log\frac{x}{s_0e^{rT}}}\right\}  \notag
\\
&\left(1+O\left((\log\frac{x}{s_0e^{rT}})^{-\frac{1}{2}}\right)\right)
\label{E:asym1}
\end{align}
as $x\rightarrow\infty$. Next, using (\ref{E:sic2}) and (\ref{E:asym1}), we
obtain 
\begin{align}
&D_T(x)=\frac{\widetilde{A}2^{\frac{n_1-3}{4}}} {T^{\frac{n_1+1}{8}}(8%
\widetilde{C}+T)^{\frac{n_1+1}{8}}}\left(s_0e^{rT}\right) ^{\frac{1}{2}+%
\frac{\sqrt{8\widetilde{C}+T}}{2\sqrt{T}}}\exp\left\{\frac{\widetilde{B}^2}{%
2(8\widetilde{C}+T)}\right\}  \notag \\
&\left(\log\frac{x}{s_0e^{rT}}\right)^{\frac{n_1-3}{4}}x^{-\left(\frac{3}{2}+%
\frac{\sqrt{8\widetilde{C}+T}}{2\sqrt{T}}\right)}\exp\left\{\frac{\widetilde{%
B}\sqrt{2}}{T^{\frac{1}{4}}(8\widetilde{C}+T)^{\frac{1}{4}}}\sqrt{\log\frac{x%
}{s_0e^{rT}}}\right\}  \notag \\
&\left(1+O\left((\log\frac{x}{s_0e^{rT}})^{-\frac{1}{2}}\right)\right)
\label{E:asym11}
\end{align}
as $x\rightarrow\infty$.

The next assertion can be obtained by using (\ref{E:fitch}) in (\ref%
{E:asym11}) and simplifying the resulting expressions.

\begin{theorem}
\label{T:tt} If $\delta> 0$, then 
\begin{align}
&D_T(x)=V\left(\log\frac{x}{s_0e^{rT}}\right)^{\frac{n_1-3}{4}}x^{-\left(%
\frac{3}{2}+\frac{\sqrt{4+\lambda_1}}{2\sqrt{\lambda_1}}\right)}\exp\left\{%
\frac{\sqrt{2\delta} }{\lambda_1^{\frac{1}{4}}(4+\lambda_1)^{\frac{1}{4}}}%
\sqrt{\log\frac{x}{s_0e^{rT}}}\right\}  \notag \\
&\left(1+O\left((\log\frac{x}{s_0e^{rT}})^{-\frac{1}{2}}\right)\right)
\label{E:vd}
\end{align}
as $x\rightarrow\infty$, where 
\begin{align}
&V=\frac{2^{\frac{n_1-5}{4}}}{\sqrt{\pi} \lambda_1^{\frac{n_1+1}{8}%
}(4+\lambda_1)^{\frac{n_1+1}{8}}}\delta^{-\frac{n_1-1}{4}}
\left(s_0e^{rT}\right) ^{\frac{1}{2}+\frac{\sqrt{4+\lambda_1}}{2\sqrt{%
\lambda_1}}} \exp\left\{-\frac{\delta(3+\lambda_1)} {2(4+\lambda_1)}\right\}
\notag \\
&\exp\left\{\frac{s-\sum_{n=1}^{\infty} \delta_n^2} {2\lambda_1}%
\right\}\prod_{k> n_1}\left(\frac{\lambda_1}{\lambda_1-\lambda_k}\right)^{%
\frac{1}{2}} \exp\left\{\frac{\delta_k^2}{2(\lambda_1-\lambda_k)}\right\}.
\label{E:oo}
\end{align}
\end{theorem}

Formula (\ref{E:vd}) describes the asymptotic behavior of the asset price
density in a Gaussian stochastic volatility model in terms of the Karhunen-Lo%
\`{e}ve parameters, the initial condition $s_0$, the interest rate $r$, and
the time horizon $T$. Note that the Karhunen-Lo\`{e}ve parameters depend on $%
T$, while the constant $V$ depends on $s_0$ and $r$. We will sometimes use
the notation $V(s_0,r)$ to emphasize this dependence.

An asymptotic formula similar to that in (\ref{E:vd}) can be also obtained
for $\delta=0$, using (\ref{E:finna}) and (\ref{E:ho}) instead of (\ref%
{E:finn}). We will next formulate this asymptotic formula for a special
model where the volatility is described by a centered Gaussian process.

\begin{theorem}
\label{T:ttt} If the process $X$ is centered, then 
\begin{align}
&D_T(x)=U\left(\log\frac{x}{s_0e^{rT}}\right)^{\frac{n_1-2}{2}}x^{-\left(%
\frac{3}{2}+\frac{\sqrt{4+\lambda_1}}{2\sqrt{\lambda_1}}\right)}\left(1+O%
\left((\log\frac{x}{s_0e^{rT}})^{-\frac{1}{2}}\right)\right)  \label{E:vdo}
\end{align}
as $x\rightarrow\infty$, where 
\begin{align}
&U=\frac{1}{\Gamma\left(\frac{n_1}{2}\right)\lambda_1^{\frac{n_1}{4}%
}(4+\lambda_1)^{\frac{n_1}{4}}} (s_0e^{rT})^{\frac{1}{2}+\frac{\sqrt{%
4+\lambda_1}}{2\sqrt{\lambda_1}}}\prod_{k> n_1}^{\infty}\left(\frac{\lambda_1%
}{\lambda_1-\lambda_k}\right)^{\frac{1}{2}}.  \label{E:oi}
\end{align}
\end{theorem}

\section{Asymptotics of the implied volatility\label{IV}}

Taking into account formula (\ref{E:DD}), we see that the discounted asset
price process in a Gaussian stochastic volatility model is given by the
following stochastic exponential: 
\begin{equation}
\widetilde{S}_t=e^{-rt}S_t=s_0\exp\left\{-\frac{1}{2}\int_0^tX_s^2ds+%
\int_0^t|X_s|dW_s\right\}.  \label{E:mart}
\end{equation}
The next standard assertion states that Gaussian stochastic volaitility
models create a risk-neutral environment.

\begin{lemma}
\label{L:martin} \label{L:environ} Under the restrictions on the volatility
process $X$ in (\ref{mart}), the discounted asset price process $\widetilde{S%
}$ is a $\{\mathcal{F}_t\}$-martingale.
\end{lemma}

\emph{Proof.} Lemma \ref{L:environ} is standard. Using It$\hat{\mathrm{o}}$%
's formula, we first show that the process $\widetilde{S}$ in (\ref{E:mart})
is a positive local martingale. Hence, it is a supermartingale by Fatou's
lemma. The conditional distribution of the stochastic integral $%
\int_{0}^{t}|X_{s}|dW_{s}$ given $|X|$ is normal with mean zero and variance 
$\int_{0}^{t}X_{s}^{2}ds$. Hence by conditioning on $|X|$ and using the
normal MGF, we can prove that $\mathbb{E}[\widetilde{S}_{t}]=s_{0}$ for all $%
t$. However, a supermartingale with a constant expectation is a martingale.
This completes the proof of Lemma \ref{L:martin}. \hspace*{\fill}$\square 
\vspace{0.1in}$

Let us define the call pricing function in the stochastic volatility model
described by (\ref{mart}) by $C(T,K)=e^{-rt}\mathbb{E}\left[(S_T-K)^{+}%
\right], $ where $T$ is the maturity and $K$ is the strike price, and recall
that $S_0=s_0$ a.s.

If the initial condition for the volatility process $X$ is constant, then
the call pricing function $C$ is free of static arbitrage. On the other
hand, if the initial condition $X_0$ is random, then there may be static
arbitrage in the function $C$. We refer the reader to Definition 1.2 in \cite%
{Ro} for more details concerning static arbitrage.

Let us fix the maturity $T$, and consider $C$ as the function $K\mapsto C(K)$
of only the strike price K. The Black-Scholes implied volatility associated
with the pricing function $C$ will be denoted by $I$. More information on
the implied volatility can be found in \cite{Ga,G}.

The asymptotic behavior of the implied volatility for stochastic volatility
models, in which the asset price density satisfies 
\begin{equation}
D_T(x)=r_1x^{-r_3}\exp\{r_2\sqrt{\log x}\}(\log x)^{r_4}(1+O\left((\log x)^{-%
\frac{1}{2}}\right), \quad x\rightarrow\infty,  \label{E:co}
\end{equation}
where $r_1> 0$, $r_2\ge 0$, $r_3> 2$, and $r_4\in\mathbb{R}$, was
characterized in \cite{GV}. However, there is an error in the expression for
the fourth coefficient in formula (91) in \cite{GV}. The correct statement
is as follows.

\begin{theorem}
\label{T:GV} Suppose condition (\ref{E:co}) holds. Then the following
asymptotic formula is valid for the implied volatility: 
\begin{align}
&I(K)=\frac{\sqrt{2}}{\sqrt{T}}(\sqrt{r_3-1}-\sqrt{r_3-2})\sqrt{\log \frac{K%
}{s_0e^{rT}}}+\frac{r_2}{\sqrt{2T}}\left(\frac{1}{\sqrt{r_3-2}} -\frac{1}{%
\sqrt{r_3-1}}\right)  \notag \\
&+\frac{2r_4+1}{2\sqrt{2T}}\left(\frac{1}{\sqrt{r_3-2}}-\frac{1}{\sqrt{r_3-1}%
}\right) \frac{\log\log\frac{K}{s_0e^{rT}}}{\sqrt{\log\frac{K}{s_0e^{rT}}}} 
\notag \\
&+\left[\frac{1}{\sqrt{2T}}\left(\frac{1}{\sqrt{r_3-1}}-\frac{1}{\sqrt{r_3-2}%
}\right) \log\frac{\sqrt{r_3-1}-\sqrt{r_3-2}}{2\sqrt{\pi}r_1} +\frac{r_2^2}{4%
\sqrt{2T}}\left(\frac{1}{(r_3-2)^{\frac{3}{2}}}-\frac{1}{(r_3-1)^{\frac{3}{2}%
}}\right)\right]  \notag \\
&\times\frac{1}{\sqrt{\log\frac{K}{s_0e^{rT}}}} +\frac{r_2(2r_4+1)}{4\sqrt{2T%
}}\left(\frac{1}{(r_3-2)^{\frac{3}{2}}}-\frac{1}{(r_3-1)^{\frac{3}{2}}}%
\right)\frac{\log\log\frac{K}{s_0e^{rT}}}{\log\frac{K}{s_0e^{rT}}} +O\left(%
\frac{1}{\log\frac{K}{s_0e^{rT}}}\right)  \label{E:rr2}
\end{align}
as $K\rightarrow\infty$.
\end{theorem}

The proof of Theorem \ref{T:GV} is exactly the same as that of Theorem 17 in 
\cite{GV}.

The next assertions (Theorems \ref{T:is} and \ref{T:isos}) are the main
results of the present paper. They provide asymptotic formulas for the
implied volatility in the stochastic volatility model given by (\ref{mart}).

\begin{theorem}
\label{T:is} Suppose $\delta> 0$. Then the following formula holds for the
implied volatility $K\mapsto I(K)$: 
\begin{align}
& I(K)=M_{1}\sqrt{\log \frac{K}{s_{0}e^{rT}}}+M_{2}+M_{3}\frac{\log \log 
\frac{K}{s_{0}e^{rT}}}{\sqrt{\log \frac{K}{s_{0}e^{rT}}}}  \notag \\
& \quad +M_{4}\frac{1}{\sqrt{\log \frac{K}{s_{0}e^{rT}}}}+M_{5}\frac{\log
\log \frac{K}{s_{0}e^{rT}}}{\log \frac{K}{s_{0}e^{rT}}}+O\left( \frac{1}{%
\log \frac{K}{s_{0}e^{rT}}}\right)  \label{E:rur2}
\end{align}%
as $K\rightarrow \infty $, where 
\begin{align}
& M_{1}=\frac{\sqrt{2}}{\sqrt{T}}\left( \frac{\sqrt{\lambda _{1}}}{\sqrt{%
4+\lambda _{1}}+2}\right) ^{\frac{1}{2}},\quad M_{2}=\frac{\sqrt{\delta }}{%
\sqrt{T}}\left( \frac{\lambda _{1}}{\sqrt{4+\lambda _{1}}(\sqrt{4+\lambda
_{1}}+2)}\right) ^{\frac{1}{2}},  \label{E:n11} \\
& M_{3}=\frac{n_{1}-1}{4\sqrt{2T}}\left( \frac{\lambda _{1}^{\frac{3}{2}}}{%
\sqrt{4+\lambda _{1}}+2}\right) ^{\frac{1}{2}},  \notag \\
& M_{4}=-\frac{1}{\sqrt{2T}}\left( \frac{\lambda _{1}^{\frac{3}{2}}}{\sqrt{%
4+\lambda _{1}}+2}\right) ^{\frac{1}{2}}\log \left[ \frac{1}{2\sqrt{\pi }%
V(1,0)}\left( \frac{\lambda _{1}^{\frac{1}{2}}}{\sqrt{4+\lambda _{1}}+2}%
\right) ^{\frac{1}{2}}\right]  \notag \\
& \quad +\frac{\sqrt{2}\delta }{4\sqrt{T}}\left( \frac{\sqrt{\lambda _{1}}(%
\sqrt{4+\lambda _{1}}-2)}{4+\lambda _{1}}\right) ^{\frac{1}{2}}(\sqrt{%
4+\lambda _{1}}+1),  \notag \\
& M_{5}=\frac{(n_{1}-1)\sqrt{\delta }}{8\sqrt{T}}\left( \frac{\lambda _{1}(%
\sqrt{4+\lambda _{1}}-2)}{\sqrt{4+\lambda _{1}}}\right) ^{\frac{1}{2}}(\sqrt{%
4+\lambda _{1}}+1),  \notag
\end{align}%
where $V\left( 1,0\right) $ is the value of $V$ in (\ref{E:oo}) with $%
s_{0}=1 $ and $r=0$.
\end{theorem}

\textit{Proof.} Set $r_1=V(1,0)$, $r_2=\frac{\sqrt{2\delta}}{\lambda_1^{%
\frac{1}{4}}(4+\lambda_1)^{\frac{1}{4}}}$, $r_3=\frac{3}{2}+\frac{\sqrt{%
4+\lambda_1}}{2\sqrt{\lambda_1}}$, and $r_4=\frac{n_1-3}{4}$. Next, using (%
\ref{E:vd}) and (\ref{E:rr2}), and making straightforward simplifications,
we get 
\begin{align*}
&M_1=\frac{2\lambda_1^{\frac{1}{4}}}{\sqrt{T}\left[(\sqrt{4+\lambda_1}+\sqrt{%
\lambda_1})^{\frac{1}{2}} +(\sqrt{4+\lambda_1}-\sqrt{\lambda_1})^{\frac{1}{2}%
}\right]}, \\
&M_2=\frac{\sqrt{2\delta\lambda_1}}{(4+\lambda_1)^{\frac{1}{4}}\sqrt{T}\left[%
(\sqrt{4+\lambda_1}+\sqrt{\lambda_1})^{\frac{1}{2}} +(\sqrt{4+\lambda_1}-%
\sqrt{\lambda_1})^{\frac{1}{2}}\right]}, \\
&M_3=\frac{(n_1-1)\lambda_1^{\frac{3}{4}}}{4\sqrt{T}\left[(\sqrt{4+\lambda_1}%
+\sqrt{\lambda_1})^{\frac{1}{2}} +(\sqrt{4+\lambda_1}-\sqrt{\lambda_1})^{%
\frac{1}{2}}\right]}, \\
&M_4=-\frac{\lambda_1^{\frac{3}{4}}}{\sqrt{T}\left[(\sqrt{4+\lambda_1}+\sqrt{%
\lambda_1})^{\frac{1}{2}} +(\sqrt{4+\lambda_1}-\sqrt{\lambda_1})^{\frac{1}{2}%
}\right]} \\
&\quad\times\log\frac{\lambda_1^{\frac{1}{4}}}{\sqrt{2\pi}V(1,0)\left[(\sqrt{%
4+\lambda_1}+\sqrt{\lambda_1})^{\frac{1}{2}} +(\sqrt{4+\lambda_1}-\sqrt{%
\lambda_1})^{\frac{1}{2}}\right]} \\
&\quad+\frac{\delta\lambda_1^{\frac{1}{4}}}{8\sqrt{T(4+\lambda_1)}}\left[(%
\sqrt{4+\lambda_1}+\sqrt{\lambda_1}) ^{\frac{3}{2}}-(\sqrt{4+\lambda_1}-%
\sqrt{\lambda_1})^{\frac{3}{2}}\right], \\
&M_5=\frac{\sqrt{2\lambda_1\delta}(n_1-1)}{32\sqrt{T}(4+\lambda_1)^{\frac{1}{%
4}}}\left[(\sqrt{4+\lambda_1}+\sqrt{\lambda_1})^{\frac{3}{2}}-(\sqrt{%
4+\lambda_1}-\sqrt{\lambda_1})^{\frac{3}{2}}\right].
\end{align*}

Finally, by taking into account the equalitites 
\begin{align*}
& (\sqrt{4+\lambda _{1}}+\sqrt{\lambda _{1}})^{\frac{1}{2}}+(\sqrt{4+\lambda
_{1}}-\sqrt{\lambda _{1}})^{\frac{1}{2}}=\sqrt{2}(\sqrt{4+\lambda _{1}}+2)^{%
\frac{1}{2}}, \\
& (\sqrt{4+\lambda _{1}}+\sqrt{\lambda _{1}})^{\frac{1}{2}}-(\sqrt{4+\lambda
_{1}}-\sqrt{\lambda _{1}})^{\frac{1}{2}}=\sqrt{2}(\sqrt{4+\lambda _{1}}-2)^{%
\frac{1}{2}}, \\
& (\sqrt{4+\lambda _{1}}+\sqrt{\lambda _{1}})^{\frac{3}{2}}-(\sqrt{4+\lambda
_{1}}-\sqrt{\lambda _{1}})^{\frac{3}{2}}=2^{\frac{3}{2}}(\sqrt{4+\lambda _{1}%
}-2)^{\frac{1}{2}}(\sqrt{4+\lambda _{1}}+1),
\end{align*}%
we obtain the formulas for the coefficients in Theorem \ref{T:is}. \hspace*{%
\fill}$\square \vspace{0.1in}$

The constant $V(1,0)$, given by (\ref{E:oo}), depends on all the Karhunen-Lo%
\`{e}ve parameters. However, this constant appears for the first time in the
fourth term of the asymptotic expansion in (\ref{E:rr2}). By keeping only
three terms in (\ref{E:rur2}), we obtain an asymptotic formula for the
implied volatility, in which the coefficients do not depend on $V$. However,
now we have the error term of the following form: $O\left(\left(\log\frac{K}{%
s_0e^{rT}}\right)^{-\frac{1}{2}}\right)$.

We will next suppose that the volatility is a centered Gaussian process, and
study the wing behavior of the implied volatility in such a case. According
to formula (\ref{E:vdo}), we can take $r_1=U(1,0)$, $r_2=0$, $r_3=\frac{3}{2}%
+\frac{\sqrt{4+\lambda_1}}{2\sqrt{\lambda_1}}$, and $r_4=\frac{n_1-2}{2}$.
Here $U(1,0)$ is defined by (\ref{E:oi}). Then, using Theorems \ref{T:ttt}
and \ref{T:GV}, and reasoning as in the proof of Theorem \ref{T:is}, we
obtain the following assertion.

\begin{theorem}
\label{T:isos} Suppose the volatility is modeled by a centered Gaussian
process. Then 
\begin{align*}
&I(K)=L_1\sqrt{\log \frac{K}{s_0e^{rT}}}+L_2 \frac{\log\log\frac{K}{s_0e^{rT}%
}}{\sqrt{\log\frac{K}{s_0e^{rT}}}} +L_3\frac{1}{\sqrt{\log\frac{K}{s_0e^{rT}}%
}} +O\left(\frac{1}{\log\frac{K}{s_0e^{rT}}}\right)
\end{align*}
as $K\rightarrow\infty$, where 
\begin{align*}
&L_1=\frac{\sqrt{2}}{\sqrt{T}}\left(\frac{\sqrt{\lambda_1}}{\sqrt{4+\lambda_1%
}+2}\right)^{\frac{1}{2}},\quad L_2=\frac{n_1-1}{2\sqrt{2T}}\left(\frac{%
\lambda_1^{\frac{3}{2}}}{\sqrt{4+\lambda_1}+2}\right)^{\frac{1}{2}}, \\
&L_3=-\frac{1}{\sqrt{2T}}\left(\frac{\lambda_1^{\frac{3}{2}}}{\sqrt{%
4+\lambda_1}+2}\right)^{\frac{1}{2}} \log\left[\frac{1}{2\sqrt{\pi}U(1,0)}%
\left(\frac{\lambda_1^{\frac{1}{2}}}{\sqrt{4+\lambda_1}+2}\right)^{\frac{1}{2%
}}\right].
\end{align*}
\end{theorem}

\begin{remark}
\label{R:symmetric} \textrm{\ Since the processes $X$ and $W$ in (\ref{mart}%
) are independent, the model in (\ref{mart}) belongs to the class of the
so-called symmetric models (see Section 9.8 in \cite{G}). It is known that
for a symmetric model, 
\begin{equation}
I(K)=I\left( \frac{\left( s_{0}e^{rT}\right) ^{2}}{K}\right) \quad 
\mbox{for
all}\quad K>0.  \label{E:ee}
\end{equation}%
It is clear that, using (\ref{E:ee}) and Theorem \ref{T:is}, we can
characterize the left-wing asymptotic behavior of the implied volatility in
the case of a noncentered Gaussian volatility. Similarly, (\ref{E:ee}) and
Theorem \ref{T:isos} can be used in the case of a centered Gaussian
volatility. }
\end{remark}

\section{Implied volatility in the uncorrelated Stein-Stein model}

\label{S:uSS} The classical Stein-Stein model is an important special
example of a Gaussian stochastic volatility model. The Stein-Stein model was
introduced in \cite{SS}. The volatility in the uncorrelated Stein-Stein
model is the absolute value of an Ornstein-Uhlenbeck process with a constant
initial condition $m_0$. In this section, we also consider a generalization
of the Stein-Stein model, in which the initial condition for the volatility
process is a random variable $X_0$. Of our interest in the present section
is a Gaussian stochastic volatility model with the process $X$ satisfying
the equation $dX_t=q(m-X_t)dt+\sigma dZ_t$. Here $q> 0$, $m\ge 0$, and $%
\sigma> 0$. It will be assumed that the initial condition $X_0$ is a
Gaussian random variable with mean $m_0$ and variance $\sigma_0^2$,
independent of the process $Z$. It is known that 
\begin{equation}
X_t=e^{-qt}X_0+(1-e^{-qt})m+\sigma e^{-qt}\int_0^te^{qu}dZ_u,\quad t\ge 0.
\label{E:ss}
\end{equation}
If $\sigma_0=0$, then the initial condition is equal to the constant $m_0$.
The mean function of the process $X$ is given by 
\begin{equation}
m(t)=e^{-qt}m_0+(1-e^{-qt})m,  \label{E:eq}
\end{equation}
and its covariance function is as follows: 
\begin{equation*}
Q(t,s)=e^{-q(t+s)}\left\{\sigma_0^2+\frac{\sigma^2}{2q}\left(e^{2q%
\min(t,s)}-1\right)\right\}.
\end{equation*}
Therefore, the following formula holds for the variance function: 
\begin{equation*}
\sigma^2_t=\frac{\sigma^2}{2q}+e^{-2qt}\left(\sigma_0^2-\frac{\sigma^2}{2q}%
\right),
\end{equation*}
and hence, if $\sigma_0^2=\frac{\sigma^2}{2q}$, then the process $X_t-m(t)$, 
$t\in[0,T]$, is centered and stationary. In this case, the covariance
function is given by 
\begin{equation*}
Q(t,s)=\frac{\sigma^2}{2q}e^{-q|t-s|}.
\end{equation*}

The Karhunen-Lo\`{e}ve expansion of the Ornstein-Uhlenbeck process is known
explicitly (see \cite{CP}). Denote by $w_n$ the increasingly sorted sequence
of the positive solutions to the equation 
\begin{equation}
\sigma^2w\cos(wT)+(q\sigma^2-w^2\sigma_0^2-q^2\sigma_0^2)\sin(wT)=0.
\label{E:eqin1}
\end{equation}
If $\sigma_0=0$, then the equation in (\ref{E:eqin1}) becomes 
\begin{equation}
w\cos(wT)+q\sin(wT)=0.  \label{E:eqin2}
\end{equation}
For the OU process in (\ref{E:ss}) with $\sigma_0\neq 0$, we have $n_k=1$
for all $k\ge 1$; 
\begin{equation}
\lambda_n=\frac{\sigma^2}{w_n^2+q^2}  \label{E:eqin}
\end{equation}
for all $n\ge 1$; and 
\begin{equation}
e_n(t)=K_n[\sigma_0^2w_n\cos(w_nt)+(\sigma^2-q\sigma_0^2)\sin(w_nt)]
\label{E:eqin3}
\end{equation}
for all $n\ge 1$ and $t\in[0,T]$. The constant $K_n$ in (\ref{E:eqin3}) is
determined from 
\begin{align}
\frac{1}{K_n^2}&=\frac{1}{2w_n}\sigma_0^2(\sigma^2-q\sigma_0^2)(1-%
\cos(2w_nT)) +\frac{1}{2}\sigma_0^4w_n^2\left(T+\frac{1}{2w_n}%
\sin(2w_nT)\right)  \notag \\
&\quad+\frac{1}{2}(\sigma^2-q\sigma_0^2)^2\left(T-\frac{1}{2w_n}%
\sin(2w_nT)\right)  \label{E:K}
\end{align}
for all $n\ge 1$. On the other hand, if $\sigma_0=0$, then $\lambda_n$ is
given by (\ref{E:eqin}), while the functions $e_n$ are defined by 
\begin{equation}
e_n(t)=\frac{1}{\sqrt{\frac{T}{2}-\frac{\sin(2w_nT)}{4w_n}}}\sin(w_nt)
\label{E:eqin4}
\end{equation}
for all $n\ge 1$ and $t\in[0,T]$.

By the Karhunen-Lo\`{e}ve theorem, the Ornstein-Uhlenbeck process $X$ in (%
\ref{E:ss}) can be represented as follows: 
\begin{equation*}
X_t=e^{-qt}m_0+(1-e^{-qt})m+\sum_{n=1}^{\infty}\sqrt{\lambda_n}e_n(t)Z_n
\end{equation*}
where $\{Z_n\}_{n\ge 1}$ is an i.i.d. sequence of standard normal variables.
The eigenvalues $\lambda_n$, $n\ge 1$, and the eigenfunctions $e_n$, $n\ge 1$%
, are given by (\ref{E:eqin}) and (\ref{E:eqin3}) if $\sigma_0\neq 0$, and
by (\ref{E:eqin}) and (\ref{E:eqin4}) if $\sigma_0=0$. Recall that the
numbers $w_n$, $n\ge 1$, in (\ref{E:eqin}) are solutions to the equation in (%
\ref{E:eqin1}) if $\sigma_0\neq 0$, and to the equation in (\ref{E:eqin2})
if $\sigma_0=0$. We refer the interested reader to \cite{CP} for more
details.

Our next goal is to discuss the constants in the asymptotic formulas for the
implied volatility at extreme strikes in the Stein-Stein model. Since $%
n_{1}=1$ for any OU process, the third and fifth terms in the expansion of
Theorem \ref{T:is} are zero, and with the exception of the term $V\left(
1,0\right) $ in $M_{4}$, the only parameters needed to compute the
above-mentioned constants are $\lambda _{1}$ and $\delta _{1}$. If $\sigma
_{0}\neq 0$, then we have 
\begin{equation}
\lambda _{1}=\frac{\sigma ^{2}}{w_{1}^{2}+q^{2}},  \label{E:lambda}
\end{equation}%
where $w_{1}$ is the smallest strictly positive solution to the equation in (%
\ref{E:eqin1}).

The next assertion provides explicit formulas for the number $%
\delta_1=\int_0^Tm(t)e_1(t)dt. $

\begin{lemma}
\label{L:lemma} (i)\,\,For the generalized uncorrelated Stein-Stein model
with $\sigma_0\neq 0$, 
\begin{align}
\delta_1&=\frac{K_1m(\sigma^2-q\sigma_0^2)(1-\cos(w_1T))}{w_1}%
+K_1\sigma_0^2\sin(w_1T)[(m_0-m)e^{-qT}+m]  \notag \\
&\quad+K_1\sigma^2(m_0-m) \frac{w_1[1-e^{-qT}\cos(w_1T)]-qe^{-qT}\sin(w_1T)}{%
q^2+w_1^2},  \label{E:tsh1}
\end{align}
where the constant $K_1$ is determined from (\ref{E:K}) with $n=1$. The
symbol $w_1$ in (\ref{E:tsh1}) stands for the smallest strictly positive
solution to (\ref{E:eqin1}). \newline
\newline
(ii)\,\,For the uncorrelated Stein-Stein model with $X_0=m_0\,\,\mathbb{P}$%
-almost surely, 
\begin{align}
&\delta_1=\frac{mq^2(1-\cos(w_1T))+w_1^2(m_0-m\cos(w_1T))} {w_1(q^2+w_1^2)%
\sqrt{\frac{T}{2}-\frac{\sin(2w_1T)}{4w_1}}}.  \label{E:tsh}
\end{align}
\end{lemma}

\emph{Proof.} Taking into account (\ref{E:eq}) and (\ref{E:eqin3}), we see
that 
\begin{align}
&\delta_1=b_1\int_0^T\cos(w_1t)dt+b_2\int_0^Te^{-qt}\cos(w_1t)dt  \notag \\
&\quad+b_3\int_0^T\sin(w_1t)dt +b_4\int_0^Te^{-qt}\sin(w_1t)dt,
\label{E:cos1}
\end{align}
where 
\begin{align}
&b_1=mK_1\sigma_0^2w_1,\quad b_2=(m_0-m)K_1\sigma_0^2w_1,  \notag \\
&b_3=mK_1(\sigma^2-q\sigma_0^2),\quad\mbox{and}\quad
b_4=(m_0-m)K_1(\sigma^2-q\sigma_0^2).  \label{E:coe}
\end{align}
It remains to evaluate the integrals in (\ref{E:cos1}). We have 
\begin{equation}
\int_0^T\cos(w_1t)dt=\frac{\sin(w_1T)}{w_1},  \label{E:i1}
\end{equation}
\begin{equation}
\int_0^Te^{-qt}\cos(w_1t)dt=\frac{q[1-e^{-qT}\cos(w_1T)]+w_1e^{-qT}\sin(w_1T)%
}{q^2+w_1^2},  \label{E:i2}
\end{equation}
\begin{equation}
\int_0^T\sin(w_1t)dt=\frac{1-\cos(w_1T)}{w_1},  \label{E:i3}
\end{equation}
and 
\begin{equation}
\int_0^Te^{-qt}\sin(w_1t)dt=\frac{w_1[1-e^{-qT}\cos(w_1T)]-qe^{-qT}\sin(w_1T)%
}{q^2+w_1^2}.  \label{E:i4}
\end{equation}
In the proof of (\ref{E:i2}) and (\ref{E:i4}), we use the integration by
parts formula twice. Now, taking into account formulas (\ref{E:cos1}-\ref%
{E:i4}) and making simplifications, we establish formula (\ref{E:tsh1}).

Next, suppose $\sigma _{0}=0$. Then (\ref{E:tsh1}) implies that 
\begin{align*}
\delta _{1}& =\frac{m}{\sqrt{\frac{T}{2}-\frac{\sin (2w_{1}T)}{4w_{1}}}}%
\frac{1-\cos (w_{1}T)}{w_{1}} \\
& \quad +\frac{m_{0}-m}{\sqrt{\frac{T}{2}-\frac{\sin (2w_{1}T)}{4w_{1}}}}%
\frac{w_{1}[1-e^{-qT}\cos (w_{1}T)]-qe^{-qT}\sin (w_{1}T)}{q^{2}+w_{1}^{2}},
\end{align*}%
where $w_{1}$ denotes the smallest strictly positive solution to (\ref%
{E:eqin2}). It is not hard to see, using the equality $w_{1}\cos
(w_{1}T)+q\sin (w_{1}T)=0$, that 
\begin{equation}
\delta _{1}=\frac{m}{\sqrt{\frac{T}{2}-\frac{\sin (2w_{1}T)}{4w_{1}}}}\frac{%
1-\cos (w_{1}T)}{w_{1}}+\frac{m_{0}-m}{\sqrt{\frac{T}{2}-\frac{\sin (2w_{1}T)%
}{4w_{1}}}}\frac{w_{1}}{q^{2}+w_{1}^{2}},  \label{E:i5}
\end{equation}%
and it is clear that (\ref{E:i5}) and (\ref{E:tsh}) are equivalent. This
completes the proof of Lemma \ref{L:lemma}. \hspace*{\fill}$\square \vspace{%
0.1in}$

\begin{remark}
\label{R:finn} \textrm{\ Since for the generalized Stein-Stein model with a
random initial condition we have $n_{1}=1$, one can use the asymptotic
formulas in Theorem \ref{T:is} with }$\mathrm{\mathrm{M}_{3}=M_{5}=0}$%
\textrm{$\mathrm{\ }$to characterize the wing-behavior of the implied
volatility. The dependence of the parameters $\lambda _{1}$ and $\delta _{1}$%
, appearing in those formulas, on the model parameters is described in (\ref%
{E:lambda}), (\ref{E:tsh1}), and (\ref{E:tsh}). Originally, sharp asymptotic
formulas for the implied volatility at extreme strikes in the uncorrelated
Stein-Stein model with $X_{0}=m_{0}$ were obtained in \cite{GS} (see also
Section 10.5 in \cite{G}). However, explicit expressions, obtained in \cite%
{GS} and \cite{G} for the coefficients in the asymptotic formulas for the
implied volatility in the Stein-Stein model, are significantly more
complicated than those found in the present paper. }
\end{remark}

\section{Numerical illustration\label{NUM}}

A basic calibration strategy when presented with asymptotic results such as
those given in this paper is to assume one can place oneself in the
corresponding regime, and then determine model parameters by reading
asymptotic coefficient off of market option prices. We now illustrate how
this strategy can produce positive results, and discuss its limitations,
when the top of the KL spectrum is simple ($n_{1}=1$). As noted in the
introduction, in this case, the third and fifth terms in the expansion are
null. The idea is to ignore the big $O$ term in the asymptotic (\ref{E:n11}%
), and calibrate parameters to the remaining coefficients. Denoting the
discounted log-moneyness $\log \left( S_{0}e^{rT}/K\right) $ by $k$ for
compactness of notation, we thus have, for $\left\vert k\right\vert $
sufficiently large,%
\begin{equation}
I\left( k\right) \simeq M_{1}\sqrt{\left\vert k\right\vert }+M_{2}+M_{4}%
\frac{1}{\sqrt{\left\vert k\right\vert }},  \label{E:onemore1}
\end{equation}%
for three constants $M_{1}$, $M_{2}$, and $M_{4}$, which can, in principle,
be read off of market data. By the explicit expressions for the first two
constants in (\ref{E:n11}) in terms of $\lambda _{1}$ and $\delta _{1}$, we
then express the latter in terms of $M_{1}$ and $M_{2}$ as 
\begin{eqnarray}
\lambda _{1} &=&\frac{64T^{2}M_{1}^{4}}{(4-T^{2}M_{1}^{4})^{2}},  \notag \\
\delta _{1} &=&\frac{4\sqrt{2T}M_{2}\sqrt{4+T^{2}M_{1}^{4}}}{4-T^{2}M_{1}^{4}%
}.  \label{calibration}
\end{eqnarray}%
Here we use (\ref{E:onemore1}). One notices that, conveniently, $\lambda
_{1} $ can be calibrated using only the coefficient $M_{1}$, while given $%
M_{1}$, $\delta _{1}$ is then proportional to $M_{2}$.

At this stage, one may simply conclude that the extreme strike asymptotics
given in the market are consistent with any Gaussian volatility model whose
top of eigenstructure is represented by the values computed in the above
expressions for $\lambda _{1}$ and $\delta _{1}$. However, practitioners
will prefer to determine a more specific model, perhaps by choosing a
classical parametric one, and using other non-asymptotic-calibration
techniques for estimating some of its parameters. The expressions in (\ref%
{calibration}) can then be used to pin down other parameters by calibration,
as long as one can relate the model's parameters to the pair $\left( \lambda
_{1},\delta _{1}\right) $ from the top of its KL spectrum, whether
analytically or numerically. The expression for $M_{4}$, given in (\ref{E:oo}%
) and (\ref{E:n11}), may be too complex to provide a reliable method for
calibrating parameters beyond the pair $\left( \lambda _{1},\delta
_{1}\right) $, but we will see below that the existence of the corresponding
term in the expansion, combined with a truncation of the formula for $M_{4}$%
, is very helpful for implementing the calibration based on (\ref%
{calibration}).

We provide illustrations of this strategy in two cases: the stationary
Stein-Stein model, where the KL expansion is known semi-explicitly, and the
Stein-Stein model's long-memory version, where the volatility is also known
as the fractional Ornstein-Uhlenbeck (fOU) model, and the KL expansion is
computed numerically. The data we use is also generated numerically: for
each model, we compute option prices and their corresponding implied
volatilities, by classical Monte-Carlo, given that the underlying pair of
stochastic processes is readily simulated. Specifically, in the Stein-Stein
(standard OU) case, $10^{6}$ paths are generated via Euler's method based on
discretizing the stochastic differential equation satisfied by $X$ started
from a r.v. sampled from $X$'s stationary distribution, and the explicit
expression for $\log S$ given $X$, also approximated via Euler with the same
time steps; $10^{3}$ time steps are used in $[0,T]$ for the various values
of $T$ we illustrate below ($1,2,3$ and $6$ months, measured in years).
Option prices are derived by computing average payoffs over the $10^{6}$
paths. The details are well known, and are omitted. In the fOU case, the
exact same methodology is used, except that one must specify the technique
used to simulate increments of the fBm process which drives $X$: we used the
circulant method, which is based on fBm's spectral properties, and was
proposed by A.T. Dieker in a 2002 thesis: see \cite{D, DMa}.

Given this simulated data, before embarking on the task of calibrating
parameters, to ensure that our methodology is relevant in practice, it is
important to discuss liquidity issues. It is known that the out-of-the-money
call options market is poorly liquid, implying that the large strike
asymptotics for call and IV prices are typically not visible in the data. We
concentrate instead on small strike asymptotics. There, depending on the
market segment, options with three-month maturity can be liquid with small
bid-ask spread for log moneyness $k$ as far down as $-1$ or even a bit
further. Options with six-month maturity with very small bid-ask spread can
be liquid as far down as $-1.5$. Convincing visual evidence of this can be
found in Figures 3 and 4 in \cite{GJ} which report 2011 data for SPX
options. We will also consider examples with one-month and two-month
maturity, where liquidity will be assumed to exist down to $k=-0.8$, based
on corresponding evidence in the same figures. We will illustrate
calibration using intervals of relatively short length which start on the
left side within these observed liquidity ranges. Beyond these lower bounds,
liquidity is insufficient to measure IV. In these ranges of $k$, the
constant term $M_{2}$ and the expressions $\sqrt{-k}$ and $1/\sqrt{-k}$ are
of similar magnitude, which may call into question whether the expansion can
be of any use in the range where liquidity exists. However, one may expect
that the KL expansion converges fast enough that the three terms $M_{1}\sqrt{%
-k}$, $M_{2}$, and $M_{4}/\sqrt{-k}$ are of different orders because the
three constants are. This turns out to be the case in the two example
classes we consider, so that our three-term expansion allows us to calibrate 
$\lambda _{1}$ and $\delta _{1}$ to $M_{1}$ and $M_{2}$ as in (\ref%
{calibration}). This works very well in practice, as our examples below now
show.\bigskip

We begin with the stationary uncorrelated Stein-Stein model with constant
mean-reversion level $m$, rate of mean reversion $q$, and so called vol-vol
parameter $\sigma $. Referring to the notation in Section \ref{S:uSS}, since
now $X$ is stationary, we have $m_{0}=m$ and $\sigma _{0}^{2}=\sigma
^{2}/\left( 2q\right) $, and the constant $K_{1}$, which is determined from
equation (\ref{E:K}), will play an important role for us. The systems of
equations needed to perform calibration here have a somewhat triangular
structure. According to Section \ref{S:uSS}, if one were to calibrate $q$,
access to $\delta _{1}$ would be needed, if one were to rely on independent
knowledge of the level of mean reversion $m$. Specifically, one would solve
the following system%
\begin{eqnarray}
q\sin (wT)+w\cos (wT) &=&0  \label{qw} \\
C_{1}\left( \sin (wT)+\frac{q}{w}(1-\cos (wT))\right) &=&\frac{\delta _{1}}{m%
}  \notag
\end{eqnarray}%
where $C_{1}=K_{1}\sigma _{0}^{2}$. As noted via (\ref{E:K}), unfortunately
the constant $C_{1}$ also depends on $\left( q,w\right) $ in the following
non-trivial way:%
\begin{equation}
\frac{1}{C_{1}^{2}}=\frac{q}{2}(1-\cos (2wT))+\frac{w^{2}}{2}\left( T+\frac{%
\sin (2wT)}{2w}\right) +\frac{q^{2}}{2}\left( T-\frac{\sin (2wT)}{2w}\right)
.  \label{C1}
\end{equation}%
When $q$ is not fixed, the task of determining which value of $w$ represents
the minimal solution of the first equation above, given the large number of
solutions to the above system, is difficult. We did not pursue this avenue
further for this reason. Instead, we will assume that $q$, which determines
the rate of mean reversion, is known, and we will calibrate the pair $\left(
m,\sigma \right) $.

The equations for finding $\left( \sigma ,m\right) $ given prior knowledge
of $q$, and given measurement of $M_{1}$ and $M_{2}$ which imply values of $%
\left( \lambda _{1},\delta _{1}\right) $ via (\ref{calibration}), are much
simpler. Indeed, since $q$ is assumed given, the base frequency $w$ is
computed easily as the smallest positive solution of (\ref{E:eqin2}). Then
according to equation (\ref{E:lambda}) and part (ii) of Lemma \ref{L:lemma},
with $C_{1}$ given by (\ref{C1}), we obtain immediately%
\begin{eqnarray}
\sigma ^{2} &=&\lambda _{1}\left( w^{2}+q^{2}\right) ;  \label{calibsigma} \\
m &=&\frac{\delta _{1}}{C_{1}\left( \sin (wT)+\frac{q}{w}(1-\cos
(wT))\right) }.  \label{calibm}
\end{eqnarray}

Any fitting method can in principle be used to estimate the coefficients $%
M_{1}$, $M_{2}$, and $M_{4}$ when working from a data-based IV curve.
However, it turns out that, in the ranges of liquidity which we described
above, any estimation will contain a certain amount of instability. We now
give the details of an iterative technique which increases the stability of
the method dramatically by exploiting the fact that $M_{4}$ is significantly
smaller than $M_{1}$ and $M_{2}$.

We use simulated IV data for the call option with $m=0.2$ (signifying a
typical mean level of volatility of $20\%$), $q=7$ (fast mean reversion,
every eight weeks or so), and $\sigma =1.2$ (high level of volatility
uncertainty). How to estimate $M_{1}$ from the data is not unambiguous. We
adopt a least-squares method, on an interval of $k$-values of fixed length;
after experimentation, as a rule of thumb, an interval of length $0.10$ or $%
0.20$ provides a good balance between providing a local estimate and drawing
on enough datapoints. One should start the interval as far to the left as
possible while avoiding any range with insufficient liquidity in practice.
As a guide to assess this liquidity, we use the study reported in \cite[%
Section 5.4]{GJ}, which depends heavily on the option maturity, as we
mentioned in this section. The following are intervals employed.\bigskip

\begin{center}
\hspace*{-0.5in}%
\begin{tabular}{|l|l|l|l|l|}
\hline
Maturity $T$ & $1/12$ (1 mo.) & $1/12$ (1 mo.) & $1/6$ (2 mos.) & $1/6$ (2
mos.) \\ \hline
Interval used & $[-0.8,-0.6]$ & $[-0.7,-0.6]$ & $[-0.8,-0.6]$ & $[-0.7,-0.6]$
\\ \hline
\end{tabular}%
\\[0pt]
\hspace*{-0.5in}%
\begin{tabular}{|l|l|l|l|l|}
\hline
Maturity $T$ & $0.25$ (3 mos.) & $0.25$ (3 mos.) & $0.5$ (6 mos.) & $0.5$ (6
mos.) \\ \hline
Interval used & $[-1.1,-0.9]$ & $[-1.0,-0.9]$ & $[-1.4,-1.2]$ & $[-1.3,-1.2]$
\\ \hline
\end{tabular}%
\bigskip
\end{center}

Graphs of the data versus the asymptotic curve in (\ref{E:onemore1}),
showing excellent agreement, are given in Fig. 1a thru 1d, though a
case-by-case need for an analysis of the trade-off between this agreement
and the liquidity-dictated calibration intervals, is apparent as one
considers various possible maturities (note the difference in ranges for
log-moneyness $k$ on the horizontal axes).

Our stabilized calibration method starts with a least-squares measurement of 
$M_{1}$ and $M_{2}$ based on the asymptotic curve with only the first two
terms. The value of $M_{1}$ is used to calibrate $\sigma $. A guess is
expressed for $m$ to initiate the procedure; in our examples we use $m=0.22,$
to signify an educated guess which misses the mark by $10\%$, as would be
reasonable to expect when using a proxy such as the VIX to visually estimate
this so-called mean reversion level. The next step uses the values of $%
\sigma $ and $m$ previously determined, along with the known value $q$, to
compute a large number of terms in the KL expansion of the OU process (we
use 500 terms), and uses those terms to compute $M_{4}$ via the expressions
in (\ref{E:oo}) and (\ref{E:n11}). The value of $M_{4}$ just obtained is
also used to refine the non-linear least-squares estimation of $M_{1}$ and $%
M_{2}$ based on the three-term asymptotic function in (\ref{E:onemore1})
where the term $M_{4}/\sqrt{\left\vert k\right\vert }$ is assumed known. The
third step then calibrates $\sigma $ and $m$ based on the new values of $%
M_{1}$ and $M_{2}$, and then recomputes $M_{4}$ using the same procedure as
in the second step, which allows a new estimation of $M_{1}$ and $M_{2}$
using the full asymptotics including the just-updated term $M_{4}/\sqrt{%
\left\vert k\right\vert }$. One then repeats the third step iteratively,
until one notices a stabilization. In the examples we report, the method
either stabilizes on a single set of values for the pair $\left( \sigma
,m\right) $, or loops between two very close sets of values; this occurs
after 6 or 7 steps. We think this needed number of repeats, and the
precision obtained in the end, are typical, because they are functions of
the small magnitude of $M_{4}$ compared to $M_{1}$ (order of $2\%$ to $10\%$
for our maturities from one month to six months), this $M_{4}$ being
considered as a nuisance term whose rough estimation helps sharpen the
estimation of the other two constants significantly. Summarizing the
procedure, we have:

\begin{enumerate}
\item[(0)] Assume $q$ is known. Compute $w$ as smallest frequency solving (%
\ref{qw}).
\end{enumerate}

\begin{enumerate}
\item Use two-term asymptotics to estimate $M_{1}$ and $M_{2}$, calibrate $%
\sigma $ to $M_{1}$ via (\ref{calibration}) and (\ref{calibm}). Initialize $%
m $ using a good guess for rate of mean reversion.

\item Use $\sigma $ and $m$ from step 1 (and $q$ from step 0) to compute a
large number (e.g. $500$) of terms in the KL expansion of $X$. Use truncated
theoretical formula in (\ref{E:oo}) and (\ref{E:n11}) to compute $K_{4}$
from this expansion. Re-estimate $M_{1}$ and $M_{2}$ by using full
three-term asymptotics (\ref{E:onemore1}) assuming $M_{4}/\sqrt{\left\vert
k\right\vert }$ is known.

\item Calibrate $\sigma $ from the new $M_{1}$ and $m$ from the new pair $%
\left( M_{1},M_{2}\right) $ via (\ref{calibration}), (\ref{calibm}), and (%
\ref{calibsigma}). Recompute the KL expansion of $X$ based on the new $%
\left( \sigma ,m\right) $, and recompute $K_{4}$ using the new KL expansion
in the theoretical formula. Re-estimate $M_{1}$ and $M_{2}$ by using full
three-term asymptotics (\ref{E:onemore1}) assuming $M_{4}/\sqrt{\left\vert
k\right\vert }$ is known using the new $M_{4}$.

\item Repeat step 3 iteratively until stabilization of $\left( \sigma
,m\right) $ occurs.
\end{enumerate}

We report our findings for the calibration of $\left( \sigma ,m\right) $ in
our 8 examples of interest in the following tables. The "true values" of $%
M_{1}$, $M_{2}$, and $M_{4}$ in these tables are those which are computed
from the Stein-Stein model with $\left( \sigma ,m,q\right) =\left(
1.2,0.2,7\right) $ via its KL elements; as explained above, only the
first-order KL eigen-elements are needed for $M_{1}$, $M_{2}$, while for $%
M_{4}$, we use the full theoretical formula in which we ignore the
eigen-elements after rank 500.\bigskip

\begin{center}
\begin{tabular}{ll}
$T=1/12$ (1 mo.) & Calibration over the interval $[-0.8,-0.6]$%
\end{tabular}

\begin{tabular}{|l|l|l|l|l|l|}
\hline
& $M_{1}$ & $M_{2}$ & $M_{4}$ & $\sigma $ & $m$ \\ \hline
True values & 0.7117 & 0.0706 & 0.0188 & 1.2 & 0.2 \\ \hline
Step 1 & 0.6875 & 0.1113 &  & 1.1196 & 0.22 \\ \hline
Step 2 & 0.6875 & 0.0777 & 0.0188 &  &  \\ \hline
Step 3 & 0.7145 & 0.0661 & 0.0183 & 1.2096 & 0.1873 \\ \hline
Step 4 & 0.7138 & 0.0673 & 0.0184 & 1.2072 & 0.1907 \\ \hline
Step 5 & 0.7140 & 0.0671 & 0.0184 & 1.2077 & 0.1900 \\ \hline
\end{tabular}%
\bigskip

\begin{tabular}{ll}
$T=1/12$ (1 mo.) & Calibration over the interval $[-0.7,-0.6]$%
\end{tabular}

\begin{tabular}{|l|l|l|l|l|l|}
\hline
& $M_{1}$ & $M_{2}$ & $M_{4}$ & $\sigma $ & $m$ \\ \hline
True values & 0.7117 & 0.0706 & 0.0188 & 1.2 & 0.2 \\ \hline
Step 1 & 0.6859 & 0.1126 &  & 1.1142 & 0.22 \\ \hline
Step 2 & 0.6859 & 0.0777 & 0.0187 & 1.2102 & 0.1872 \\ \hline
Step 3 & 0.7147 & 0.0661 & 0.0183 & 1.2102 & 0.1872 \\ \hline
Step 4 & 0.7141 & 0.0671 & 0.0184 & 1.2081 & 0.1901 \\ \hline
Step 5 & 0.7143 & 0.0668 & 0.0184 & 1.2086 & 0.1894 \\ \hline
\end{tabular}%
\bigskip

\begin{tabular}{ll}
$T=1/6$ (2 mos.) & Calibration over the interval $[-0.8,-0.6]$%
\end{tabular}

\begin{tabular}{|l|l|l|l|l|l|}
\hline
& $M_{1}$ & $M_{2}$ & $M_{4}$ & $\sigma $ & $m$ \\ \hline
True values & 0.5743 & 0.0704 & 0.0245 & 1.2 & 0.2 \\ \hline
Step 1 & 0.5370 & 0.1309 &  & 1.0490 & 0.22 \\ \hline
Step 2 & 0.5370 & 0.0775 & 0.0232 &  &  \\ \hline
Step 3 & 0.5705 & 0.0752 & 0.0251 & 1.1839 & 0.2134 \\ \hline
Step 4 & 0.5732 & 0.0706 & 0.0245 & 1.1953 & 0.2005 \\ \hline
Step 5 & 0.5723 & 0.0720 & 0.0247 & 1.1917 & 0.2046 \\ \hline
Step 6 & 0.5726 & 0.0716 & 0.0246 & 1.1929 & 0.2032 \\ \hline
Step 7 & 0.5725 & 0.0718 & 0.0246 & 1.1923 & 0.2039 \\ \hline
\end{tabular}%
\bigskip \newpage

\begin{tabular}{ll}
$T=1/6$ (2 mos.) & Calibration over the interval $[-0.7,-0.6]$%
\end{tabular}

\begin{tabular}{|l|l|l|l|l|l|}
\hline
& $M_{1}$ & $M_{2}$ & $M_{4}$ & $\sigma $ & $m$ \\ \hline
True values & 0.5743 & 0.0704 & 0.0245 & 1.2 & 0.2 \\ \hline
Step 1 & 0.5354 & 0.1322 &  & 1.0428 & 0.22 \\ \hline
Step 2 & 0.5354 & 0.0775 & 0.0231 &  &  \\ \hline
Step 3 & 0.5711 & 0.0748 & 0.0251 & 1.1866 & 0.2123 \\ \hline
Step 4 & 0.5742 & 0.0698 & 0.0244 & 1.1995 & 0.1982 \\ \hline
Step 5 & 0.5731 & 0.0715 & 0.0246 & 1.2011 & 0.1997 \\ \hline
Step 6 & 0.5734 & 0.0710 & 0.0246 & 1.1962 & 0.2017 \\ \hline
\end{tabular}%
\bigskip

\begin{tabular}{ll}
$T=1/4$ (3 mos.) & Calibration over the interval $[-1.1,-0.9]$%
\end{tabular}

\begin{tabular}{|l|l|l|l|l|l|}
\hline
& $M_{1}$ & $M_{2}$ & $M_{4}$ & $\sigma $ & $m$ \\ \hline
True values & 0.5001 & 0.0702 & 0.0295 & 1.2 & 0.2 \\ \hline
Step 1 & 0.4699 & 0.1299 &  & 1.0591 & 0.22 \\ \hline
Step 2 & 0.4699 & 0.0773 & 0.0279 &  &  \\ \hline
Step 3 & 0.4980 & 0.0740 & 0.0300 & 1.1896 & 0.2107 \\ \hline
Step 4 & 0.5001 & 0.0698 & 0.0294 & 1.1997 & 0.1987 \\ \hline
Step 5 & 0.4995 & 0.0710 & 0.0295 & 1.1969 & 0.2021 \\ \hline
Step 6 & 0.4996 & 0.0708 & 0.0295 & 1.1973 & 0.2016 \\ \hline
\end{tabular}%
\bigskip

\begin{tabular}{ll}
$T=1/4$ (3 mos.) & Calibration over the interval $[-1.0,-0.9]$%
\end{tabular}

\begin{tabular}{|l|l|l|l|l|l|}
\hline
& $M_{1}$ & $M_{2}$ & $M_{4}$ & $\sigma $ & $m$ \\ \hline
True values & 0.5001 & 0.0702 & 0.0295 & 1.2 & 0.2 \\ \hline
Step 1 & 0.4655 & 0.1342 &  & 1.0396 & 0.22 \\ \hline
Step 2 & 0.4655 & 0.0773 & 0.0275 &  &  \\ \hline
Step 3 & 0.4945 & 0.0777 & 0.0305 & 1.1732 & 0.2214 \\ \hline
Step 4 & 0.4977 & 0.0716 & 0.0295 & 1.1883 & 0.2039 \\ \hline
Step 5 & 0.4966 & 0.0736 & 0.0298 & 1.1832 & 0.2097 \\ \hline
Step 6 & 0.4969 & 0.0730 & 0.0297 & 1.1847 & 0.2080 \\ \hline
Step 7 & 0.4968 & 0.0732 & 0.0297 & 1.1842 & 0.2086 \\ \hline
\end{tabular}%
\bigskip

\begin{tabular}{ll}
$T=1/2$ (6 mos.) & Calibration over the interval $[-1.4,-1.2]$%
\end{tabular}

\begin{tabular}{|l|l|l|l|l|l|}
\hline
& $M_{1}$ & $M_{2}$ & $M_{4}$ & $\sigma $ & $m$ \\ \hline
True values & 0.3838 & 0.0695 & 0.0428 & 1.2 & 0.2 \\ \hline
Step 1 & 0.3521 & 0.1432 &  & 1.0094 & 0.22 \\ \hline
Step 2 & 0.3521 & 0.0765 & 0.0385 &  &  \\ \hline
Step 3 & 0.3817 & 0.0757 & 0.0442 & 1.1869 & 0.2178 \\ \hline
Step 4 & 0.3861 & 0.0657 & 0.0423 & 1.2144 & 0.1890 \\ \hline
Step 5 & 0.3846 & 0.0690 & 0.0429 & 1.2052 & 0.1986 \\ \hline
Step 6 & 0.3851 & 0.0679 & 0.0427 & 1.2081 & 0.1956 \\ \hline
Step 7 & 0.3849 & 0.0683 & 0.0428 & 1.2071 & 0.1966 \\ \hline
Step 8 & 0.3850 & 0.0681 & 0.0427 & 1.2076 & 0.1961 \\ \hline
\end{tabular}%
\bigskip \newpage

\begin{tabular}{ll}
$T=1/2$ (6 mos.) & Calibration over the interval $[-1.3,-1.2]$%
\end{tabular}

\begin{tabular}{|l|l|l|l|l|l|}
\hline
& $M_{1}$ & $M_{2}$ & $M_{4}$ & $\sigma $ & $m$ \\ \hline
True values & 0.3838 & 0.0695 & \multicolumn{1}{|l|}{0.0428} & 1.2 & 0.2 \\ 
\hline
Step 1 & 0.3493 & 0.1464 &  & 0.9934 & 0.22 \\ \hline
Step 2 & 0.3493 & 0.0765 & 0.0380 &  &  \\ \hline
Step 3 & 0.3797 & 0.0783 & 0.0446 & 1.1745 & 0.2255 \\ \hline
Step 4 & 0.3850 & 0.0665 & 0.0423 & 1.2075 & 0.1915 \\ \hline
Step 5 & 0.3832 & 0.0706 & 0.0430 & 1.1959 & 0.2034 \\ \hline
Step 6 & 0.3837 & 0.0694 & 0.0428 & 1.1994 & 0.1998 \\ \hline
Step 7 & 0.3836 & 0.0697 & 0.0429 & 1.1984 & 0.2008 \\ \hline
Step 8 & 0.3836 & 0.0696 & 0.0428 & 1.1989 & 0.2003 \\ \hline
\end{tabular}%
\bigskip
\end{center}

We obtain excellent agreement of the calibration with the true values, with
errors lower than $1\%$ after $5$ to $8$ steps. Other calibrations, not
reported here because of their similarity with these examples, show that
calibration accuracy would increase with more liquid options since these
allow being able to use intervals further to the left, ensuring a better
match with the asymptotic regime (\ref{E:onemore1}). The examples reported
above in full correspond to realistic liquidity assumptions.\bigskip

We now propose a calibration method to estimate the memory parameter in the
fOU volatility model. This model was introduced in \cite{CR} as a way to
model long-range dependence in non-linear functionals of asset returns,
while preserving the uncorrelated semi-martingale structure at the level of
returns themselves. This is the model for $X$ in (\ref{OU}) where the
process $Z$ is a fractional Brownian motion, i.e. the continuous Gaussian
process started at $0$ with covariance determined by $\mathbf{E}\left[
\left( Z_{t}-Z_{s}\right) ^{2}\right] =\left\vert t-s\right\vert ^{2H}$,
with \textquotedblleft Hurst\textquotedblright\ parameter $H\in (0.5,1)$. In 
\cite{CV}, it was shown empirically that standard statistical methods for
long-memory data are inadequate for estimating $H$. This difficulty can be
attributed to the fact that the volatility process $X$ can have
non-stationary increments. In addition, some of the classical methods use
path regularity or self-similarity as a proxy for long memory, which cannot
be exploited in practice since there is a lower limit to how frequenty
observations can be made without running into microstructure noise. To make
matters worse, the process $X$ is not directly observed; in such a partial
observation case, a general theoretical result was given in \cite{R}, by
which the rate of convergence of any estimator of $H$ cannot exceed an
optimal $H$-dependent rate which is always slower than $N^{-1/4}$, where $N$
is the number of observations. Given the non-stationarity of the parameter $%
H $ on a monthly scale, a realistic time series at the highest observation
frequency where microstructure noise can be ingored (e.g. one stock
observation every 5 minutes) would not permit even the optimal estimators
described in \cite{R} from pinning down a value of $H$ with any acceptable
confidence level. The work in \cite{CV} proposes a calibration technique
based on a straightforward comparison of simulated and market option prices
to determine $H$. \ Our strategy herein is similar, but based on implied
volatility.

Our goal is to calibrate the fOU model described above with the following
parameters: $T=1/4$, $m=0.2$, $q=7$, $\sigma =1.2$, with different values of
the Hurst parameter $H$, namely 
\begin{equation}
H\in \{0.51,0.55,0.60,0.65,0.70,0.75,0.80,0.85\}  \label{fixedH}
\end{equation}%
As mentioned above, our simulated option prices use standard Monte Carlo,
where the fOU process is produced by A.T. Dieker's circulant method. Since
the values of $\lambda _{1}$ for each $H>0.5$ are not known explicitly or
semi-explicitly, we resorted to the method developed in by S. Corlay in \cite%
{C} for optimal quantization: there, the infinite-dimensional eigenvalue
problem is converted to a matrix eigenvalue problem which uses a low-order
quadrature rule for approximating integrals (a trapezoidal rule is
recommended), after which a Richardson-Romberg extrapolation is used to
improve accuracy. We repeat this procedure for the fOU process with the
above parameters, for each value of $H$ from $0.50$ to $0.99$, with
increments of $0.01$. The corresponding values we obtain for $\lambda _{1}$
in each case are collected in the following table:\vspace{0.12in}

$\hspace*{-0.25in}%
\begin{array}{ccccccccccc}
\text{{\small \emph{H }=}} & \text{{\small 0.50}} & \text{{\small 0.51}} & 
\text{{\small 0.52}} & \text{{\small 0.53}} & \text{{\small 0.54}} & \text{%
{\small 0.55}} & \text{{\small 0.56}} & \text{{\small 0.57}} & \text{{\small %
0.58}} & \text{{\small 0.59}} \\ 
{\small \lambda }_{1}\text{=} & \text{{\small 0.0157}} & {\small 0.0155} & 
{\small 0.0152} & {\small 0.0150} & {\small 0.0148} & {\small 0.0146} & 
{\small 0.0144} & {\small 0.0142} & {\small 0.0140} & {\small 0.0138}%
\end{array}%
$

$\hspace*{-0.25in}%
\begin{array}{lllllllllll}
\text{{\small \emph{H }=}} & \text{{\small 0.60}} & \text{{\small 0.61}} & 
\text{{\small 0.62}} & \text{{\small 0.63}} & \text{{\small 0.64}} & \text{%
{\small 0.65}} & \text{{\small 0.66}} & \text{{\small 0.67}} & \text{{\small %
0.68}} & \text{{\small 0.69}} \\ 
{\small \lambda }_{1}\text{=} & {\small 0.0136} & {\small 0.0134} & {\small %
0.0132} & {\small 0.0130} & {\small 0.0128} & {\small 0.0126} & {\small %
0.0124} & {\small 0.0122} & {\small 0.0120} & {\small 0.0118}%
\end{array}%
$

$\hspace*{-0.25in}%
\begin{array}{lllllllllll}
\text{{\small \emph{H }=}} & \text{{\small 0.70}} & \text{{\small 0.71}} & 
\text{{\small 0.72}} & \text{{\small 0.73}} & \text{{\small 0.74}} & \text{%
{\small 0.75}} & \text{{\small 0.76}} & \text{{\small 0.77}} & \text{{\small %
0.78}} & \text{{\small 0.79}} \\ 
{\small \lambda }_{1}\text{=} & {\small 0.0116} & {\small 0.0115} & {\small %
0.0113} & {\small 0.0111} & {\small 0.0109} & {\small 0.0108} & {\small %
0.0106} & {\small 0.0104} & {\small 0.0103} & {\small 0.0101}%
\end{array}%
$

$\hspace*{-0.25in}%
\begin{array}{lllllllllll}
\text{{\small \emph{H }=}} & \text{{\small 0.80}} & \text{{\small 0.81}} & 
\text{{\small 0.82}} & \text{{\small 0.83}} & \text{{\small 0.84}} & \text{%
{\small 0.85}} & \text{{\small 0.86}} & \text{{\small 0.87}} & \text{{\small %
0.88}} & \text{{\small 0.89}} \\ 
{\small \lambda }_{1}\text{=} & {\small 0.0100} & {\small 0.0098} & {\small %
0.0097} & {\small 0.0095} & {\small 0.0094} & {\small 0.0092} & {\small %
0.0091} & {\small 0.0089} & {\small 0.0088} & {\small 0.0087}%
\end{array}%
$

$\hspace*{-0.25in}%
\begin{array}{lllllllllll}
\text{{\small \emph{H }=}} & \text{{\small 0.90}} & \text{{\small 0.91}} & 
\text{{\small 0.92}} & \text{{\small 0.93}} & \text{{\small 0.94}} & \text{%
{\small 0.95}} & \text{{\small 0.96}} & \text{{\small 0.97}} & \text{{\small %
0.98}} & \text{{\small 0.99}} \\ 
{\small \lambda }_{1}\text{=} & {\small 0.0085} & {\small 0.0084} & {\small %
0.0083} & {\small 0.0082} & {\small 0.0080} & {\small 0.0079} & {\small %
0.0078} & {\small 0.0077} & {\small 0.0076} & {\small 0.0075}%
\end{array}%
$

$~$

Our illustration of the calibration method then consists of starting with
simulated IV data for a fOU model with a fixed $H$ from the set in (\ref%
{fixedH}), then, similarly to what we did for the Stein-Stein model,
calibrate the value of $\lambda _{1}$ to the first term of the simulated IV
curve over an interval of length $0.1$. For our choice of $T=1/4$ we use $%
k\in \lbrack -1.0,-0.9]$ to determine $\lambda _{1}$, which is realistic in
terms of liquidity constraints. We then match that value of $\lambda _{1}$
to the closest value in the above table, thereby concluding that the
simulated data is consistent with the corresponding value of $H$ in the
table. Because of the instability in determining $M_{4}$ in (\ref{E:onemore1}%
) by curve fitting, as noted for the standard Stein-Stein model, rather than
using the iterative technique described above, we fit our simulated data
curve to the first two terms in this expansion only, resulting in a robust
estimate for $M_{1}$ in all cases, from which our calibrated $\lambda _{1}$
results via (\ref{calibration}). This is more efficient since we are only
calibrating the single parameter $H$. The results of this method are
summarized here.$\vspace{0.12in}$

\begin{center}
\begin{tabular}{|l|l|}
\hline
$T=1/4$ (3 mos.) & Calibration of $H$ via $\lambda _{1}$ over the interval $%
[-1.0,-0.9]$ \\ \hline
\end{tabular}

\begin{tabular}{|l|l|l|l|l|l|l|l|l|}
\hline
True $H$ & $0.51$ & $0.55$ & $0.60$ & $0.65$ & $0.70$ & $0.75$ & $0.80$ & $%
0.85$ \\ \hline
True $\lambda _{1}$ & $0.0155$ & $0.0146$ & $0.0136$ & $0.0126$ & $0.0116$ & 
$0.0108$ & $0.0100$ & $0.00923$ \\ \hline
calibrated $\lambda _{1}$ & $0.0152$ & $0.0147$ & $0.0134$ & $0.0127$ & $%
0.0115$ & $0.0109$ & $0.0101$ & $0.00937$ \\ \hline
calibrated $H$ & $0.52$ & $0.55$ & $0.61$ & $0.64$ & $0.71$ & $0.74$ & $0.79$
& $0.84$ \\ \hline
\end{tabular}%
$\vspace{0.12in}$
\end{center}

Our method shows a good level of accuracy. One notes a bias between the
curve $M_{1}\sqrt{-k}+M_{2}$ and the simulated IV data, as illustrated in
Figures 2a to 2h, which appears to shift downward as $H$ increases. Since we
are only calibrating $H$ via $\lambda _{1}$ which is inferred from $M_{1}$,
this bias has no influence on the calibration. At the cost of computing $%
M_{4}$ as we did for the Stein-Stein model, which would be more onerous in
the fOU case because one would need to push Corlay's method much further for
estimating KL eigenelements, we could obtain the 3-term expansion in (\ref%
{E:onemore1}), resulting in curves which would have much less of a bias than
in Figures 2a to 2h, but this would not improve the calibration of $\lambda
_{1}$ and $H$.



\begin{figure}[h!]
\noindent 
\includegraphics[scale=0.45]{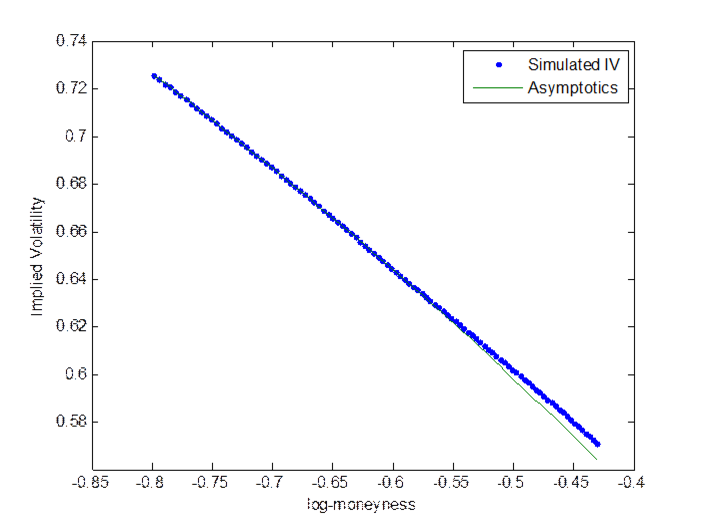}
\caption{Figure 1a. One-month IV for
Stein-Stein model with parameters $m=0.2$, $q=7$ , $\protect\sigma =1.2$}
\end{figure}

\begin{figure}[h!]
\noindent 
\includegraphics[scale=0.45]{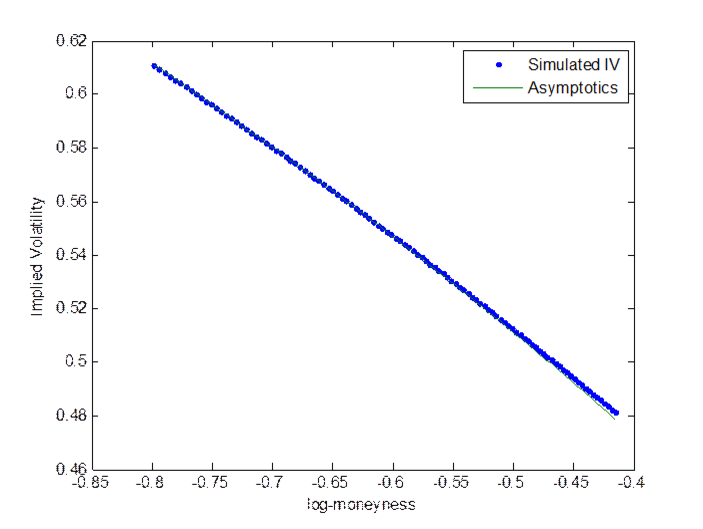}
\caption{Figure 1b. Two-month IV for Stein-Stein model with parameters $m=0.2$, $q=7$ , $\protect\sigma =1.2$}
\end{figure}

\begin{figure}[h!]
\noindent 
\includegraphics[scale=0.45]{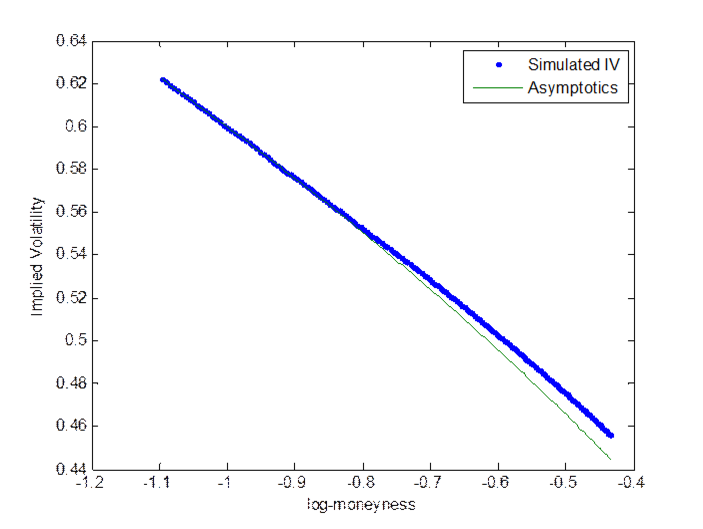}
\caption{Figure 1c. Three-month IV for
Stein-Stein model with parameters $m=0.2$, $q=7$ , $\protect\sigma =1.2$}
\end{figure}

\begin{figure}[h!]
\noindent 
\includegraphics[scale=0.45]{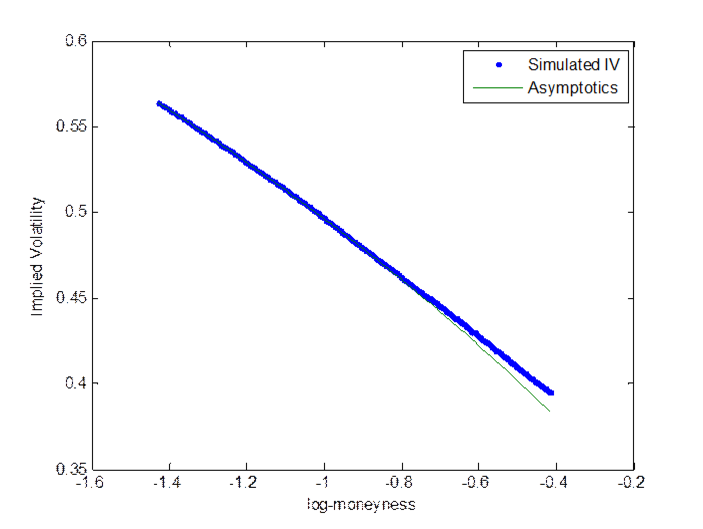}
\caption{Figure 1d. Six-month IV for Stein-Stein model with parameters $m=0.2$, $q=7$ , $\protect\sigma =1.2$}
\end{figure}

\begin{figure}[h!]
\noindent 
\includegraphics[scale=0.35]{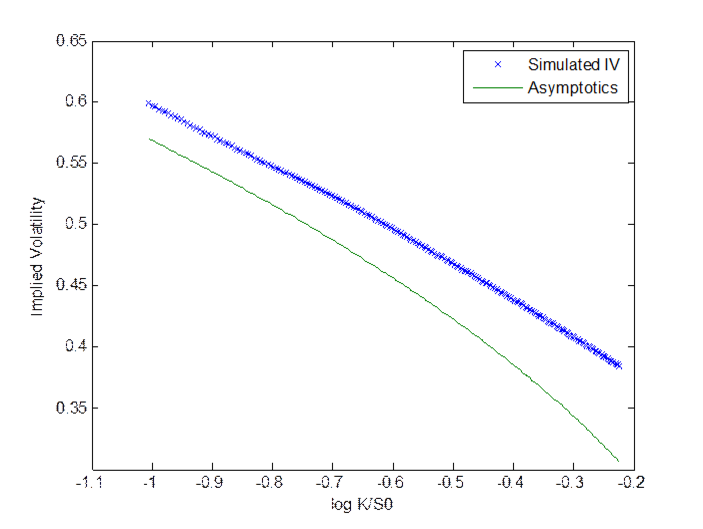}
\includegraphics[scale=0.35]{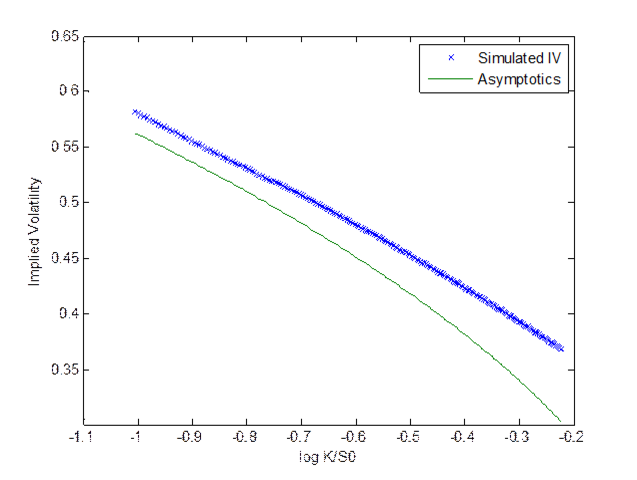}
\includegraphics[scale=0.35]{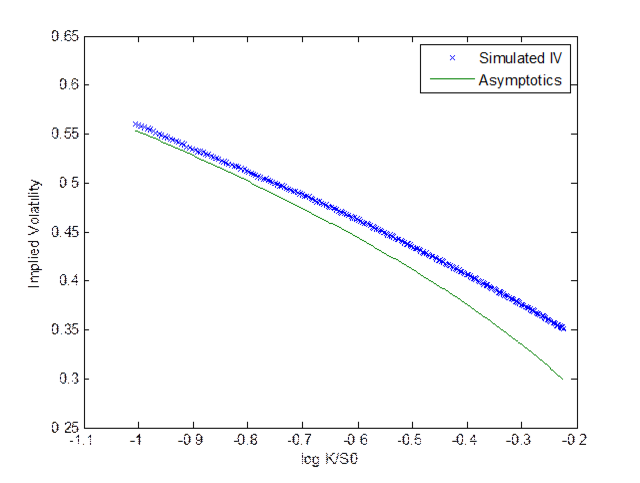}
\includegraphics[scale=0.35]{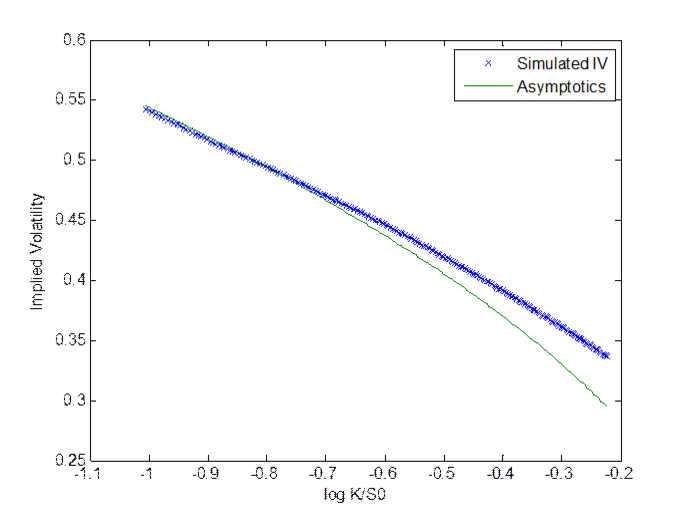}
\caption{Figures 2a, 2b, 2c, 2d. IV for fOU model
with $H=0.51$, $H=0.55$, $H=0.60$, $H=0.65$}
\end{figure}

\begin{figure}[h!]
\noindent 
\includegraphics[scale=0.34]{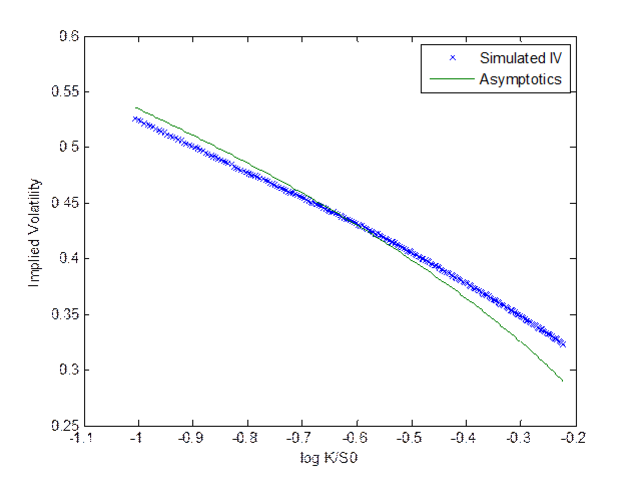}
\includegraphics[scale=0.31]{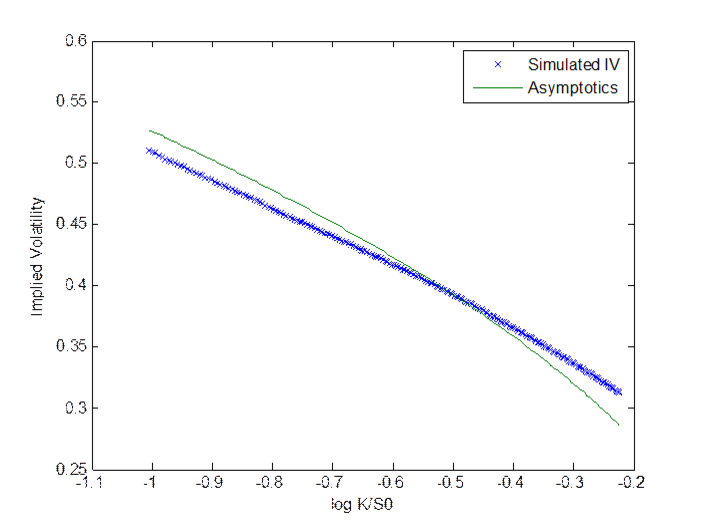}
\includegraphics[scale=0.33]{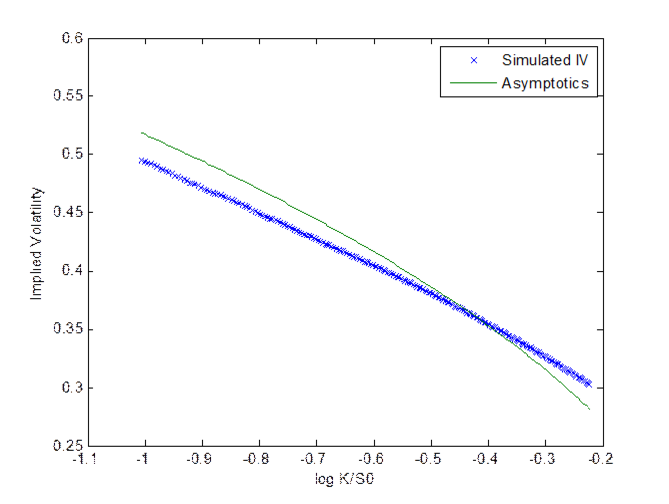}
\includegraphics[scale=0.32]{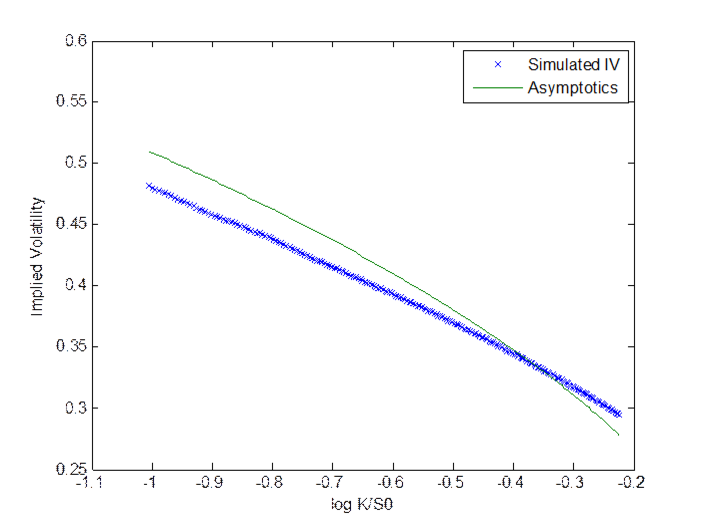}
\caption{Figures 2e, 2f, 2g, 2h. IV for fOU model
with $H=0.70$ $H=0.75$, $H=0.80$, $H=0.85$}
\end{figure}

\end{document}